\def\year{2021}\relax
\documentclass[letterpaper]{article} 
\usepackage{aaai21}  
\usepackage{times}  
\usepackage{helvet} 
\usepackage{courier}  
\usepackage[hyphens]{url}  
\usepackage{graphicx} 
\usepackage{natbib}  
\usepackage{caption} 
\usepackage{hyperref}
\usepackage{amsthm}
\usepackage{multirow}
\usepackage{algorithm}
\usepackage{amsmath}
\usepackage{algorithmic}
\usepackage{amssymb}
\usepackage{booktabs}
\usepackage{tikz}
\usepackage{mathrsfs}
\usepackage{bm}
\usepackage{subcaption}
\usepackage{comment}
\usepackage{subcaption}
\usepackage[switch]{lineno}

\newtheorem{definition}{Definition}
\newtheorem{theorem}{Theorem}
\newtheorem{axiom}{Axiom}
\newtheorem{example}{Example}
\newtheorem{lemma}{Lemma}
\newtheorem{case}{Case}
\newtheorem{subcase}{Case}[case]

\newtheorem{claim}{Claim}
\theoremstyle{definition}

\newtheorem{fact}{Fact}

\newcommand{\RN}[1]{%
  \textup{\uppercase\expandafter{\romannumeral#1}}%
}

\newcommand{\ex}{\backslash}

\newcommand{\ra}{\rightarrow}
\newcommand{\la}{\leftarrow}
\newcommand{\cb}{\lceil \beta\rceil}

\newcommand{\B}{\ensuremath{\mathsf{B}}}
\newcommand{\DB}{\ensuremath{\mathsf{DB}}}

\newcommand{\nul}{\ensuremath{0}}
\newcommand{\dum}{\ensuremath{\mathit{dum}}}
\newcommand{\g}{\ensuremath{\mathcal{G}}}
\newcommand{\V}{\ensuremath{V}}
\renewcommand{\d}{\ensuremath{\mathbf{d}}}
\newcommand{\D}{\ensuremath{\mathbf{D}}}

\newcommand{\DD}{\ensuremath{\Delta}}

\newcommand{\remove}[1]{}

\title{Power in Liquid Democracy}
\author{
Yuzhe Zhang\textsuperscript{\rm 1} and Davide Grossi\textsuperscript{\rm 1, \rm2}\\
} 

\affiliations{
    \textsuperscript{\rm 1}
    University of Groningen\\
    \textsuperscript{\rm 2}
    University of Amsterdam

}
\begin{document}

\maketitle

\begin{abstract}
The paper develops a theory of power for delegable proxy voting systems.
We define a power index able to measure the influence of both voters and delegators.
Using this index, which we characterize axiomatically, we extend an earlier game-theoretic model by incorporating power-seeking behavior by agents. We analytically study the existence of pure strategy Nash equilibria in such a model.
Finally, by means of simulations, we study the effect of relevant parameters on the emergence of power inequalities in the model.
\end{abstract}

\section{Introduction}
Liquid  democracy \citep{blum2016liquid} is a form of proxy voting \citep{miller1969program,Tullock_1992,Alger_2006,green2015direct,cohensius2017proxy} where each proxy is delegable, thereby giving rise to so-called transitive delegations. In such a system each voter may choose to cast her vote directly, or to delegate her vote to a proxy, who may in turn decide whether to vote or delegate, and so pass the votes she has accrued further to yet another proxy. The voters who decide to retain their votes---the so-called gurus---cast their ballots, which now carry the weight given by the number of delegations they accrued. Liquid democracy has been an influential proposal in recent public debates on democratic reform across the world, thanks also to platforms for democratic decision support such as, in particular, LiquidFeedback \cite{liquid_feedback}\footnote{\url{https://liquidfeedback.org/}}. In the last couple of years it has enjoyed considerable attention from researchers in political science and eDemocracy, as well as artificial intelligence (see, for an overview, \citet{paulin20overview}).

\vspace{-0.3cm}

\paragraph{Contribution}
The starting point of our paper is a controversial feature of liquid democracy: transitive delegations may in principle lead to disproportionate accrual of power, thereby harming the democratic legitimacy of the resulting vote. To the best of our knowledge, this issue has received only limited attention. A notable exception is \cite{kling15voting}, which provided an empirical analysis of power and influence in liquid democracy based on data from the German Pirate Party\footnote{\url{https://www.piratenpartei.de/}}. However, a formal theory of power in voting systems with delegable proxy is lacking. We aim at providing such a theory here, and use it to gain insights into how power may happen to be distributed among agents involved in decision-making with liquid democracy.


{\em First}, we provide a generalization of the power index known as Banzhaf index to account for delegations in voting with quota rules. This novel index---called delegative Banzhaf index---measures not only the influence of voters, but also that of delegators. 
We characterize this index axiomatically (Theorem \ref{th:char}) and highlight how the index responds intuitively to the way in which delegations may be structured (Fact \ref{fa:todelegate}-\ref{fact:accrual}).

{\em Second}, we extend the strategic model of liquid democracy developed by \citet{bloembergen19rational} to account for power-seeking behavior by agents. In our model, agents want to vote truthfully in order to relay correct information to the mechanism, but they do so by also considering how much power they retain in the system. We carry out an equilibrium analysis (pure strategy Nash equilibrium)  of the model. We show equilibria may not exist if delegations are constrained (Theorem \ref{th:noNE}), but they do when everybody is allowed to delegate to anybody (Theorem \ref{thm:NEincomplete}).

{\em Finally}, we simulate our game theoretic model and study how two key parameters of the model influence the distribution of power both in equilibrium and after one-shot interaction. Our experiments show that limiting the level of connectivity of the underlying network has a beneficial effect in limiting the emergence of inequalities in the distribution of power (measured by Gini coefficient). Perhaps less intuitively, the extent by which agents are motivated by the accumulation of power has a similar effect: groups where agents are more power-greedy appear to achieve more equal distributions of power.

\smallskip

Proofs of the simpler results are omitted from the text, but proof sketches are provided for the more interesting results. Full proofs of all results, as well as a more detailed descriptions of our experiments, are reported in the appendix.


\paragraph{Related Work}
The idea of voting with delegable proxy can be traced back to \cite{dodgson84principles} and has been object of study in the political sciences \cite{green2015direct}.
In the last couple of years, several papers in the artificial intelligence community (and in particular the computational social choice one \cite{comsoc_handbook}) have focused on liquid democracy. Two lines of research have broadly been pursued. 
On the one hand papers have pointed to potential weakenesses of voting by liquid democracy, e.g.: delegation cycles and the failure of individual rationality in multi-issue voting \cite{christoff17binary,brill2018pairwise}; poor accuracy of group decisions as compared to those achievable via direct voting in non-strategic settings \cite{kahng18liquid,caragiannis19contribution}, as well as strategic ones \cite{bloembergen19rational}. 
On the other hand a number of papers have focused on the development of better behaved delegation schemes, e.g.: delegations with preferences over trustees \cite{brill2018pairwise} or over gurus \cite{escoffier19convergence,escoffier20iterative}; multiple delegations \cite{golz2018fluid}; complex delegations like delegations to a majority of trustees \cite{colley20smart}; dampened delegations \cite{boldi2011viscous}; breadth-first delegations \cite{Grammateia2020aamas}. 

Our paper is a contribution to the first line of research.
The possibility of large power imbalances is recognized as a potential problem for liquid democracy, although experimental work has argued the issue may be limited in practice \cite{kling15voting}.
We aim at putting the discussion on power in liquid democracy on a precise footing and gain insights into how  power imbalances may arise or be contained.


\section{Preliminaries: A Model of Liquid Democracy} 
\label{sec:graph}

Our model is based on the binary voting setting for truth-tracking \cite{Condorcet1785,grofman83thirteen,elkind16rationalizations}. The setting has already been applied to the study of liquid democracy by \citet{kahng18liquid,bloembergen19rational,caragiannis19contribution}.

\vspace{-0.3cm}

\paragraph{Binary voting by truth-tracking agents}
A finite set of agents $N=\{1,2,\dots,n\}$ has to vote on whether to accept or reject an issue. The vote is supposed to track the correct state of the world---that is whether it is `best' to accept or reject the issue. The agents' ability to make the right choice (i.e., the agents' error model) is represented by the agent's \emph{accuracy} $q_i \in (\frac{1}{2},1]$, for $i\in N$. 


We assume the result of such an election to be determined by a quota rule with quota $\beta \in (\frac{n}{2}, n]$.
That is, the issue is accepted if and only if there are at least $\beta$ agents supporting it.
We will also be working with the more general setting in which each agent is endowed with a weight.
Let $\omega: N\ra \mathbb{R}_+$ be a weight function assigning a positive weight to every agent.\footnote{Therefore, in the standard `one-voter-one-vote' setting $\omega(i) = 1$ for all $i \in N$.}
Then the quota is $\beta \in \left(\frac{\sum_{i \in N} \omega(i)}{2}, \sum_{i \in N} \omega(i)\right]$. That is, an issue is accepted if and only if the weight she collects from individual votes matches or exceeds the quota (cf. \citet{chalkiadakis12computational}).

\vspace{-0.3cm}



\paragraph{Liquid democracy elections}

When agent $i\in N$ delegates to agent $j \in N$ we write $d_i = j$. We admit the possibility for an agent to abstain by delegating to a nul agent $\nul$. This feature will be of technical use for the characterization of the power index we are going to introduce.
Then $\d=(d_1,d_2,\dots, d_n)$ is called a \emph{delegation profile} (or simply \emph{profile}) and is a vector describing each agent's delegation. Equivalently, delegation profiles can be usefully thought of as maps $\d: N\rightarrow N \cup \{ \nul \}$, where $\d(i)=d_i$.
When $d_i=i$, agent $i$ votes on her own behalf. We call such an agent a \emph{guru}. On the other hand, any agent who is not a guru, is called a {\em delegator}. For profile $\d$, and $C\subseteq N$, let $C^{\d}$ denote all gurus in $C$ in profile $\d$, i.e., $C^{\d}=\{i \in C\mid d_i=i \}$. A delegation profile in which all agents are gurus (i.e., for all $i \in N$ $d_i = i$) is said to be {\em trivial}.

\smallskip

We call a {\em liquid democracy election} (LDE) the tuple $\V = \langle N, \omega, \d, \beta\rangle$, where $N$ is the set of agents with weights according to $\omega$, ${\bf d}$ is a delegation profile, and $\beta$ is the quota. Let then $\mathbb{V}$ denote the set of all LDEs. Clearly, LDEs with trivial profiles are instances of standard weighted voting.


\paragraph{Gurus, chains and cycles}
Any profile $\d$ can also be represented by a directed graph. An edge from agent $i$ to $j$ ($i\ra j$) exists whenever $d_i=j$.
Consider then a profile $\d$ where a path exists from $i$ to $j$, i.e., $i\rightarrow k \rightarrow \dots \rightarrow h \rightarrow j$.
We call such paths \emph{delegation chains}. When such a chain from $i$ to guru $j$ exists, every agent in this delegation chain (indirectly) delegates to $j$, and we denote $i$'s guru by $d^*_i = \d^*(i) = j$. 
Additionally, the set of agents between any pair of agents on the delegation chain are called the \emph{intermediaries} between the two agents.
For example, suppose the above delegation chain occurs in profile $\d$. Then the set of intermediaries between $i$ and $j$ is $\{k, \dots, h\}$, and it is denoted by $\delta_{\d}(i,j)$. The sum of the weights of the intermediaries between two agents $i$ and $j$ and the weight of $j$, is called the \emph{delegation distance} from $i$ to $j$ and is denoted by $\DD_{\d}(i,j)=\sum_{a\in(\delta_{\d}(i,j)\cup\{j\})}\omega(a)$.
A {\em delegation cycle} is a chain where the first and last agents coincide. In such a case, no agent in the chain is linked to a guru. Therefore no agent linked via a delegation chain to an agent in a delegation cycle has a guru.
Then for any $C\subseteq N$,
we write
$
\D(C) = \{ j \in N \mid \exists k\in C, d^*_j= k\} 
$
to denote the set of agents that directly or indirectly delegate to some agent in $C$. 
If $C = \{ i \}$ we write $\D(i)$ for the set of agents who have $i$ as guru. 

One last piece of notation: we will need to consider what happens to delegation chains when we restrict to certain subsets of agents.
For instance, given the chain $i\rightarrow k \rightarrow h \rightarrow \dots \rightarrow o \rightarrow j$, if $\{i,h,\dots, o,j\}\subseteq C \subseteq N$ but $k\notin C$, then $i$ is not able to delegate to $j$ within $C$ as she has no access to intermediary $k$ in such subset. For $C \subseteq N$ we write ${\bf d}_C^*(h)= j$ to denote that $j$ is the guru of $h$ and the chain from $h$ to $j$ contains only elements of $C$. Then we write
$
\hat{\D}(C) = \{ j \in N \mid \exists k\in C, {\d}_C^*(j)= k \} \label{eq:accrual2}
$ 
for the set of agents that directly or indirectly delegate to some agent in $C$ through intermediaries contained in $C$.
Intuitively, this captures the support accrued by gurus in $C$ via agents in $C$.







\section{A Power Index for Liquid Democracy}

Once delegations are settled, liquid democracy results in weighted voting where only gurus vote with the sum of weights they accrued from direct or indirect delegations.
From a voting perspective, gurus are therefore the only agents who retain voting power after the delegation phase. However, this neglects the power that delegators actually have within liquid democracy by being able to control large number of votes. By means of a simple example: a guru $i$ obtaining $m$ direct delegations is intuitively more `powerful' than a guru obtaining $m$ delegations via an intermediary $j$, who is in turn recipient of $m-1$ direct delegations. Most of $i$'s power depends then on $j$ (see also Example \ref{ex:1} below).

So in this section we generalize the Banzhaf index \cite{penrose46elementary,banzhaf65weighted} to the delegable proxy voting setting. The Banzhaf index has already been used to study the power of gurus in liquid democracy by \citet{kling15voting}.


\subsection{Delegative Banzhaf Index: Definition}

We briefly recall the definition of the Banzhaf index. A simple game is a tuple $\g = \langle N,\nu\rangle$, where $N$ is the set of agents ($|N|=n$) and $\nu$ is the characteristic function $\nu: 2^N \ra \{ 0, 1\}$. For any $C \subseteq N$, if $\nu(C) = 1$ then $C$ is said to be {\em winning}, otherwise it is said to be {\em losing}. 
An agent $i$ is called a \emph{swing agent} for coalition $C$ if $\nu(C)-\nu(C\ex \{i\})=1$.
Then in a simple game $\g$, the Banzhaf index $\B_i(\g)$ of agent $i\in N$ is:
$
\B_i(\g)=\frac{1}{2^{n-1}}\sum_{C\subseteq N \setminus \{ i \}}(\nu(C \cup \{ i \})-\nu(C)) \label{eq:banzhaf},
$
i.e., $i$'s probability of being swing for a random coalition.

\smallskip

There is one obvious way in which an LDE $\V$ induces a simple game: it is the simple game capturing the weighted voting occurring among gurus once delegations have been fixed, i.e., $\g_\V = \langle N, \nu_\V \rangle$ where, for any $C \subseteq N$: 
\begin{align}
\nu_\V(C) = 1 & \ \mbox{iff} \ \sum_{i\in \D(C)}\omega(i) \geq \beta.
\end{align}
That is, a coalition wins whenever all gurus in it together accrue enough weight to meet the quota.
In such a game only gurus may have positive power: $i \in N^\d$ if $\B_i(\g_\V) > 0$, as $\g_\V$ is silent about the influence that delegators have in determining the winning coalitions.


The influence of delegators can be captured by a different simple game $\g'_\V = \langle N, \nu'_V \rangle$ where, for any $C \subseteq N$: 
\begin{align}
\nu'_\V(C) = 1 & \ \mbox{iff} \  \sum_{i\in \hat{\D}(C)}\omega(i) \geq \beta. 
\end{align}
That is, a coalition $C$ is winning whenever the sum of weights accrued by the gurus in $C$ {\em from agents in} $C$, meets the quota.
According to this way of constructing the simple game, an agent's weight is accrued in a coalition $C$ if the agent, her guru, and all intermediaries between them are contained in $C$.
We refer to $\g'_\V$ as the {\em delegative simple game} of LDE $\V$. Clearly, if $\d$ is trivial, all agents are gurus and therefore $\g_\V = \g_{\V'}$. 

\smallskip

So, given an LDE $\V=\langle N,\omega, \d, \beta\rangle$,  we define the delegative Banzhaf index of an agent $i$ in LDE $\V$ simply as the Banzhaf index of $i$ in the delegative simple game of $\V$:
\begin{align}
\DB_i(\V) = \B_i(\g'_\V). \label{eq:db}
\end{align}
Observe that in LDEs $\V$ where the delegation profile is trivial, and therefore games $\g(\V)$ and $\g'(\V)$ coincide, the delegative Banzhaf index of each agent 
coincides to her Banzhaf index in $\g_\V$ .

\begin{example}
\label{ex:1}
Consider two LDEs, $\V_1=\langle N, \omega,\d_1,\beta\rangle$ and $\V_2=\langle N,\omega,\d_2, \beta\rangle$, where $N=\{1,2,3,4\}$, $\omega(i)=1$ for all $i\in N$, $\beta=3$ and $\d_1$ and $\d_2$ are represented in Fig.~\ref{fig:ex11} and Fig.~\ref{fig:ex12}, respectively.

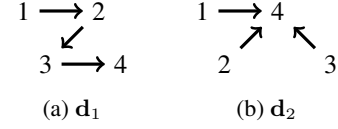
\begin{figure}[t]
\centering
\subcaptionbox
{$\d_1$\label{fig:ex11}}{\begin{tikzpicture}
\node(1){1};
\node(2)[right of=1, node distance=1cm]{2};
\node(3)[below left of=2, node distance=1cm]{3};
\node(4)[right of=3, node distance=1cm]{4};
\draw[very thick, ->](1)--(2);
\draw[very thick, ->](2)--(3);
\draw[very thick, ->](3)--(4);
\end{tikzpicture}}
\hspace{0.5cm}
\subcaptionbox
{$\d_2$\label{fig:ex12}}{\begin{tikzpicture}
\node(4){4};
\node(1)[left of=4, node distance=1cm]{1};
\node(2)[below left of=4, node distance=1cm]{2};
\node(3)[below right of=4, node distance=1cm]{3};
\draw[very thick, ->](1)--(4);
\draw[very thick, ->](2)--(4);
\draw[very thick, ->](3)--(4);
\end{tikzpicture}}
\caption{Profiles in Example~\ref{ex:1}}
\end{figure}

We focus on the indices of agents $1$ and $4$.
First consider $\V_1$.
Since no coalition $C\subseteq N$ exists with $\nu'_{\V_1}(C)=1$ and $\nu'_{\V_1}(C\setminus\{1\})=0$, $\DB_1(\V_1)=0$.
Then we compute $\DB_4(\V_1)$.
$\nu'_{\V_1}(C)=1$ and $\nu'_{\V_1}(C\setminus\{4\})=0$ iff $C\in \{\{2,3,4\},\{1,2,3,4\}\}$.
Thus $\DB_4(\V_1)=\sum_{C\subseteq N}(\nu'_{\V_1}(C)-\nu'_{\V_1}(C\setminus\{4\}))/2^3=1/4.$
Then consider $\V_2$.
First for $\DB_1(\V_2)$, $\nu'_{\V_2}(C)=1$ and $\nu'_{\V_2}(C\setminus\{1\})=0$ iff $C\in \{\{1,2,4\},\{1,3,4\}\}$.
Therefore, $\DB_1(\V_2)=\sum_{C\subseteq N}(\nu'_{\V_2}(C)-\nu'_{\V_2}(C\setminus\{1\}))/2^3=1/4.$
For $\DB_4(\V_2)$, we have that if $C\subseteq \{\{1,2,4\},\{1,3,4\},\{2,3,4\},\{1,2,3,4\}\}$, $\nu'_{\V_2}(C\cup\{4\})=1$ and $\nu'_{\V_2}(C)=0$.
Then $\DB_4(\V_2)=\sum_{C\subseteq N}(\nu'_{\V_2}(C)-\nu'_{\V_2}(C\setminus\{4\}))/2^3=1/2$.
\end{example}

As the example illustrates each agent's $\DB$ depends on the structure of the delegation profile.
For instance, in both $\V_1$ and $\V_2$, agent $1$ is a delegator with no incoming delegations, but $\DB_1(\V_1)=0$ while $\DB_1(\V_2)=1/4$.
In $\V_1$, agent $1$ is ``far" from the guru and her vote does not matter for meeting the quota.
In $\V_2$, agent $1$ delegates directly to the guru.
Similarly, in both LDEs agent $4$ collects 4 votes.
However, $\DB_4(\V_1)=1/4$ but $\DB_4(\V_2)=1/2$ since in $\V_1$, the delegation chain pointing to $4$ is long, so that agent $4$ depends on $3$ for almost all her weight.

\subsection{Characterization of $\DB$}

To underpin \eqref{eq:db} we present a characterization of the delegative Banzhaf index. We want to axiomatically identify $\DB$ among all functions
$
f: \mathbb{V} \to (N \to \mathbb{R})
$
for LDEs on $N$. To do so we borrow ideas and techniques from existing axiomatizations of the Banzhaf index for weighted voting games \cite{axiomshapley,axiomnowak,axiomlehere}. 


The strategy we follow consists in generalizing a known characterization of the Banzhaf index for standard weighted voting due to \citet{newcharacterizationbanzhaf}. We use the same axioms of that characterization (Axioms \ref{ax:maximum}-\ref{ax:sum} below), with the addition of one axiom for so-called dummy agents (Axiom \ref{ax:nul}). Crucially, however, we show how to adapt the key definitions upon which the axioms are based from the standard weighted voting setting to LDEs. This concerns in particular the definitions of composition and bloc formation (Definitions \ref{def:completeness} and \ref{def:bloc}) which play an important role in the proof. As a result one can retrieve the known characterization of the standard Banzhaf index from ours, by simply restricting to the class of LDEs where profiles are trivial, and therefore delegations do not matter.


\subsubsection{Preliminary Definitions}

We start by introducing standard definitions from the theory of simple games. Assume an LDE $\V=\langle N,\omega, \d, \beta\rangle$ be given.


\begin{definition}[Dummy Agent]
\label{df:dummy}
An agent $i\in N$ is dummy if for any $C\subseteq N$ ($i\in C$), $\nu'_\V(C)=\nu'_\V(C\setminus \{i\})$,
where $\nu'$ is the characteristic function of the delegative simple game of $\V$. Let $N^{\dum}$ denote all dummy agents.
\end{definition}
That is, an agent is dummy
whenever she cannot influence $\nu'_\V(C)$ by quitting or joining any coalition $C\subseteq N$.
It is worth observing that, in LDEs there are three ways in which an agent can be dummy: if the agent abstains (i.e., delegates to $0$); if the agent is linked by a chain to a delegation cycle; if the agent---call it $i$---is such that $\DD_\d(i,d^*_i) \ge \beta$, that is, the delegation distance between $i$ and her guru in $\d$ is larger than $\beta$. We call such an agent {\em distant} (in $\d$).


\begin{definition}[Dictator]
\label{df:dictator}
An agent $i\in N$ is a dictator if $\nu'_\V(C)=1$ if and only if $i\in C$, for any $C\subseteq N$.
\end{definition}
That is, an agent $i$ is a dictator of $\V$ whenever it belongs to all and only the winning coalitions of the delegative simple game of $\V$. In an LDE this occurs if the dictator $i$ is a guru and $\beta\in (0, \omega(i)]$, that is, $i$ meets the quota on her own.

\begin{definition}[Symmetric Agents]
\label{df:symmetric}
Any two agents $i,j\in N$ are symmetric if for all $C\subseteq N\ex \{i,j\}$, $\nu'_V(C\cup\{i\})=\nu'_V(C\cup\{j\})$.
\end{definition}
Symmetric agents are swing for exactly the same coalitions in the delegative simple game of $\V$. 
Note that a pair of symmetric agents do not necessarily have the same weight.


\begin{example}[Example~\ref{ex:1} continued]
\label{ex:2}
Consider $\V_1$ in Figure \ref{fig:ex11}. Since $\beta=3$ and the delegation distance $\Delta_{\d}(1,4)=3$, agent $1$ is a distant (and therefore dummy) agent. 
Next consider agents $1$ and $2$ (or any pair of $\{1,2,3\}$) in $\V_2$, each of whom directly delegates to agent $4$.
For any coalition $C\subseteq N\setminus \{1,2\}$, $\nu'_{\V_2}(C\cup\{1\})=\nu'_{\V_2}(C\cup\{2\})$, thus $1$ and $2$ are symmetric.
There is no dictator in Example~\ref{ex:1}.
\end{example}

The following definitions are also based on the standard theory of simple games, but are generalized in order to account for delegations.
\begin{definition}[Minimally Winning Coalition]
\label{df:minimum}
A coalition $C\subseteq N$ is a minimally winning coalition if for any $i\in \hat{\D}(C)$, $\nu'_V(C)=1$ and $\nu'_V(C\ex \{i\})=0$.
\end{definition}
That is, a coalition $C$ is minimally winning if it is winning (in the delegative simple game of $\V$), but becomes losing if any agent who is linked to a guru in $C$ via agents in $C$ is removed. So a minimally winning coalition is a coalition that contains just enough gurus with just enough support through intermediaries in the same coalition to meet the quota. It follows that no distant agent may be included in a minimally winning coalition. Notice, however, that such a coalition may contain agents that are not linked to gurus in $C$ by intermediaries in $C$ (i.e., that do not belong to $\hat{\D}(C)$) and therefore it may not be minimal w.r.t. set inclusion.


\begin{definition}[Unanimity LDE (ULDE)]
\label{df:unanimity}
$\V$ is a unanimity LDE if the quota $\beta = \sum_{i\in \D(N)}\omega(i)$. 
We call such a quota unanimity quota and denote it by $\beta^U$.
\end{definition}
That is, in a ULDE the quota equals the sum of weights of all agents who directly or indirectly delegate to gurus.



\smallskip

The last two definitions concern operations on LDEs: how to combine two LDEs into a new one; and how to build an LDE from another one by merging two agents into a so-called `bloc'.

\begin{definition}[Composition]
\label{def:completeness}
Let two LDEs $\V_1=\langle N_1,\omega_1,\d_1,\beta_1\rangle$ and $\V_2=\langle N_2,\omega_2, \d_2,\beta_2\rangle$ be given, such that for any $i\in N_1\cap N_2$, if $\d_1(i)=j$ and $j\in N_2$ (resp. $\d_2(i)=j$ and $j\in N_1$), $\d_2(i)=j$ (resp. $\d_1(i)=j$), otherwise $\d_2(i)=0$ (resp. $\d_1(i)=0$),
and $\omega_1(i)=\omega_2(i)$.
We define two new LDEs 
$\V_1\wedge \V_2=\langle N_1\cup N_2,\omega_{1\wedge 2}, \d_{1\wedge 2}, \beta_1\wedge \beta_2 \rangle$ 
and $\V_1\vee \V_2=\langle N_1\cup N_2,\omega_{1\vee 2}, \d_{1\vee 2}, \beta_1\vee \beta_2\rangle$, where:
\begin{itemize}

\item for any $i\in N_1$ (resp. $i\in N_2$), 
$\omega_{1\vee 2}(i) = \omega_{1\wedge 2}(i) = \omega_1(i)$ (resp. $\omega_{1\vee 2}(i) = \omega_{1\wedge 2}(i) = \omega_2(i)$);


\item for any $i\in N_1 \setminus N_2$ (resp. $N_2 \setminus N_1$), $\d_{1\vee 2}(i) = \d_{1\wedge 2}(i) = \d_1(i)$ (resp. $\d_{1\vee 2}(i) = \d_{1\wedge 2}(i) = \d_2(i)$), and for any $i \in N_1 \cap N_2$, if $\d_1(i) = \d_2(i)$ then $\d_{1\vee 2}(i) = \d_{1\wedge 2}(i) = \d_1(i) = \d_2(i)$, otherwise $\d_{1\vee 2}(i) = \d_{1\wedge 2}(i) = \d_k(i)$ where $k \in \{1,2\}$ and $\d_k(i) \neq 0$;

\item $\beta_1\wedge\beta_2$ (resp. $\beta_1\vee\beta_2$) is met iff $\sum_{i\in \hat{\D}_{C\cap N_1}(C\cap N_1)}\omega_1(i)\ge \beta_1$ {\em and} (resp. {\em or}) $\sum_{i\in \hat{\D}_{C\cap N_2}(C\cap N_2)}\omega_2(i)\ge \beta_2$.


\end{itemize}
\end{definition}
Two LDEs can be composed provided the delegation graphs at their intersection coincide or, if they do not, provided that this is because of one of the agents delegating outside the intersection and the other abstaining (i.e., delegating to $0$). The condition is required to guarantee the coherency of delegations in the composition. Then quotas in the composition are so defined as to guarantee that coalitions in the delegative simple game of the composition are winning iff they are winning in both, or at least one of, the delegative simple games of the LDEs (cf. proof of Lemma \ref{lemma:characterization1}).


\begin{definition}[Bloc formation] \label{def:bloc}
Given $V=\langle N,\omega, \d,\beta\rangle$ and for any $i,j\in N$ such that $d_i=j$ or $i,j\in N^{\d}$, $V'=\langle N',\omega',\d',\beta\rangle$ is called the bloc LDE joining $i$ and $j$ into a bloc $ij$, where
\begin{itemize}
\item $N'=N\ex \{i,j\}\cup \{ij\}$;
\item For $\d'$, if $d_i=j$, $d'_{ij}=d_j$, and for any $a\in N$, such that $\d_a=i$ or $\d_a=j$, $\d'_a=ij$, but if $i,j\in N^{\d}$, for any $a\in N$ such that $d_a=i$ or $d_a=j$, $d'_a=ij$;
\item $\omega'(ij)=\omega(i)+\omega(j)$. 
\end{itemize}
\end{definition}
A bloc LDE treats two agents $i$ and $j$, who are either adjacent in the delegation graph or both gurus, as one new agent $ij$.
By applying the operation in Definition~\ref{def:bloc} repeatedly, it is possible to coalesce all agents who share the same guru
into one bloc. Furthermore, any pair of delegation chains can also be joined into one bloc by joining their gurus into one bloc. Such operations play an important role in the proof of our characterization result (cf. proof of Lemma \ref{lemma:characterization2}).

\subsubsection{Axioms}
We can now introduce the axioms of our characterization.
Assume again that an LDE $\V$ is given.

We assign minimum power to dummy agents, maximum to dictators, and identical power to symmetric agents:
\begin{axiom}[No Power (\textbf{NP})]
\label{ax:nul}
If $i\in N^{dum}$, $f_i(V)=0$.
\end{axiom}
\begin{axiom}[Maximum Power (\textbf{MP})]
\label{ax:maximum}
The power index of a dictator is 1.
\end{axiom}
\begin{axiom}[Equal Treatment (\textbf{ET})]
\label{ax:equal}
For any pair of symmetric agents $i,j\in N$, $f_i(V)=f_j(V)$.
\end{axiom}

The last two axioms concern how the index should behave with respect to composition and bloc formation.
\begin{axiom}[Bloc Principle (\textbf{BP})]
\label{ax:bloc}
For any two agents $i,j\in N$ such that $d_i=j$, or $i,j\in N^{\d}$, let $V'$ be the bloc LDE by joining $i$ and $j$ into bloc $ij$. Then $f_{ij}(V') = f_i(V)+f_j(V)$.
\end{axiom}
\begin{axiom}[Sum Principle (\textbf{SP})]
\label{ax:sum}
For any pair of LDEs $V_1,V_2\in \mathbb{V}$, such that any $i\in N_1\cap N_2$ satisfies the condition in Definition~\ref{def:completeness},
$f_i(V_1\wedge V_2)+f_i(V_1\vee V_2)=f_i(V_1)+f_i(V_2)$ for any $i\in N_1\cup N_2$.
\end{axiom}
Intuitively,  Axiom~\ref{ax:bloc} requires the power of the bloc, which is formed by two gurus or two adjacent agents on a delegation chain, to be the sum of their individual powers.
Axiom~\ref{ax:sum} requires that the sum of any agent's power in $V_1\wedge V_2$ and $V_1\vee V_2$ be the sum of her power in $\V_1$ and $\V_2$.

\subsubsection{Characterization}

The result is based on two lemmas.

\begin{lemma} \label{lemma:characterization1}
A power index $f$ for liquid democracy satisfies \textbf{MP}, \textbf{NP}, \textbf{SP}, \textbf{ET}, and \textbf{BP} if it is $\DB$.
\end{lemma}
\begin{proof}[Proof sketch]
We prove the claim for \textbf{SP} as it provides a nice illustration of the workings of our definitions. The other cases are provided in the appendix. 
To show that $\DB$ satisfies the \textbf{SP}, one first has to show that by the way in which weights $\beta_1 \land \beta_2$ and $\beta_1 \vee \beta_2$ are set in Definition \ref{def:completeness}, we have that for any coalition $C\subseteq N_1\cup N_2$, $\nu'_{\V_1\wedge \V_2}(C)=1$ iff $\nu'_{\V_1}(C\cap N_1)=1$ {\em and} $\nu'_{\V_2}(C\cap N_2)=1$, and $\nu'_{\V_1\vee\V_2}(C)=1$ iff $\nu'_{\V_1}(C\cap N_1)=1$ {\em or} $\nu'_{\V_2}(C\cap N_2)=1$. The proof can then proceed with a standard argument.

We first consider any $i\in N_1-N_2$, i.e., agent $i$ is contained in $N_1$ but not in $N_2$.
Let $m_i^V$ denote the number of times that agent $i$ is swing in the delegative simple game for $\V$, i.e., $m_i^V=|\{C\subseteq N\setminus \{i\}\mid \nu'_V(C)=0, \nu'_V(C\cup \{i\})=1\}|$.
Then, if $i$ is swing in $C\subseteq N_1$ in LDE $V_1$, she is also swing in $(C\cup C')\cap N_1$, for any $C'\subseteq N_2-N_1$.
Therefore, in LDE $V_1\vee V_2$, 
$
m_i^{V_1\vee V_2}=m_i^{V_1}2^{|N_2-N_1|},
$ 
where $m_i^{V_1\vee V_2}$ is the number of times that $i$ is swing in LDE $V_1\vee V_2$.
Additionally, since $i\in N_1-N_2$, $m_i^{V_2}=0$, that is, $i$ cannot be swing in LDE $V_2$, which implies $m_i^{\V_1\wedge\V_2}=0$.
Hence we have for $i\in N_1-N_2$, 
$
m_i^{V_1\vee V_2}=m_i^{V_1}2^{|N_2-N_1|}+m_i^{V_2}2^{|N_1-N_2|}-m_i^{V_1\wedge V_2}.
$
Identical equations can be developed for agent $i\in N_2-N_1$ or $i\in N_1\cap N_2$.
We then divide each side of the equation by $2^{|N_1\cup N_2|-1}$ and obtain that, for any $i\in N_1\cup N_2$,
$
\frac{m_i^{V_1\vee V_2}}{2^{|N_1\cup N_2|-1}}=\frac{m_i^{V_1}}{2^{|N_1|-1}}+\frac{m_i^{V_2}}{2^{|N_2|-1}}-\frac{m_i^{V_1\wedge V_2}}{2^{|N_1\cup N_2|-1}},
$
which implies that $\DB_i(V_1\wedge V_2)+\DB_i(V_1\vee V_2)=\DB_i(V_1)+\DB_i(V_2)$.
\end{proof}




\begin{lemma} \label{lemma:characterization2}
A power index $f$ for LDEs satisfies \textbf{MP}, \textbf{NP}, \textbf{SP}, \textbf{ET}, and \textbf{BP}, only if it is $\DB$.
\end{lemma}
\begin{proof}[Proof sketch]
The proof consists of two claims: \fbox{{\em claim 1}} if $f$ is $\DB$ for any ULDE, it is $\DB$ for any LDE; \fbox{{\em claim 2}} $f$ is $\DB$ for any ULDE.

To prove \fbox{{\em claim 1}}, observe that any LDE $\V$ can be represented as the disjunction of $m$ ULDEs, i.e., $\V=\V_1\vee \dots, \V_m$, given $\V$ has $m$ minimally winning coalitions $\{C_1,\dots,C_m\}$, where $\V_i=\langle C_i,\omega, \d_{C_i},\beta^U\rangle$.
By induction over the number of disjunction ULDEs we show then that if $f$ is $\DB$ for the disjunction of any $k$ ($1\le k< m$) ULDEs, it is also $\DB$ for the disjunction of any $k+1$ ULDEs.
The claim holds by \textbf{SP}, since $f_i(\V_1\vee\dots\vee\V_{k+1})=f_i(\V_1\vee\dots\vee\V_k)+f_i(\V_{k+1})-f_i((\V_1\vee\dots\V_k)\wedge\V_{k+1})$, where $f$ is $\DB$ for $\V_1\vee\dots\vee\V_k$ and $\V_{k+1}$ by assumption, as well as for $(\V_1\vee\dots\V_k)\wedge\V_{k+1}$ because it is also disjunction of $k$ ULDEs: $(\V_1\vee\dots\V_k)\wedge\V_{k+1}=(\V_1\wedge\V_{k+1})\vee\dots(\V_k\wedge\V_{k+1})$.

To prove \fbox{{\em claim 2}} consider an ULDE $\V_j$. We need to show that $f_i(\V_j)=1/2^{n_j-1}$ for any non-dummy agent, and $f_i(\V_j)=0$ for any dummy agent, where $n_j=|C_j\setminus C_j^{dum}|$.
The proof is conducted by induction on $|C_j|$.
As the basis, $f_i(\V_j)=1$ if $|C_j|=1$ by \textbf{MP} since $i$ is a dictator. 
Assume then that the claim holds when $|C_j|=k$, we show it also holds if $|C_j|=k+1$. 
For any non-dummy agent, if only one non-dummy agent exists, the claim is obvious by \textbf{MP}.
If $|C_j\setminus C_j^{dum}|\ge 2$, we exploit \textbf{BP} to join two gurus, or a delegator with her trustee, into a bloc, then obtain an LDE with $k$ agents where the hypothesis holds.
Then the claim follows by \textbf{ET}.
For dummy agents the claim is proven by exploiting \textbf{NP}.
\end{proof}

\begin{theorem} \label{th:char}
A power index for liquid democracy satisfies \textbf{MP}, \textbf{NP}, \textbf{SP}, \textbf{ET}, and \textbf{BP}, if and only if it is \DB.
\end{theorem}
\begin{proof}
It follows from Lemma \ref{lemma:characterization1} and Lemma \ref{lemma:characterization2}.
\end{proof}


\subsubsection{Further Properties of $\DB$}
Besides the above axioms, it is worth mentioning a few other properties of the index that highlight its dependence on the delegation graph.


\begin{fact}[Delegation \& Power Loss]
\label{fa:todelegate}
For any pair of LDEs $\V=\langle N, \omega,\d, \beta \rangle$ and $\V'=\langle N, \omega,\d', \beta \rangle$, such that $\d'=(\d_{-i},d'_i)$, $\d(i)=i$ and $d'_i=j$ ($i\not=j$), we have that $\DB_i(\V)\ge \DB_i(\V')$.
\end{fact}
That is, delegations never lead to an increase in power for the delegator. In fact one can show that the inequality $\DB_i(\V)\ge \DB_i(\V')$ can be strict. 

The last two facts show that agents closer to the guru have more power and that, power-wise, delegating directly to a guru is better than doing that indirectly.

\begin{fact}[Power Monotonicity]
\label{fa:powermonotonicity}
For any pair of agents $i,j\in N$, such that $d_i=j$, $\DB_i(\V)\le \DB_j(\V)$.
\end{fact}

\begin{fact}[Direct vs. Indirect Delegation] \label{fact:accrual}
\label{fa:forward}
Let $\V$ be a LDE, in which there exists three agents $i,j,k\in N$, such that $\d(i)=j$ and $\d(k)=i$.
Let then $\V'=\langle N,\omega, \d',\beta\rangle$, such that $\d'=(\d_{-k},d'_k)$ where $d'_k=j$.
Then, $\DB_k(\V')\ge \DB_k(\V)$.
\end{fact}




\section{A Game-theoretic Model}
 
We will now use the \DB~index to extend the game-theoretic model of liquid democracy by~\citet{bloembergen19rational}, referred to as delegation game. Like in that model, we will be assuming that delegations are constrained by an underlying social network represented by a directed graph $\langle N, E \rangle$,
where each agent is a node in the network and for any $i,j\in N$ ($i\not= j$), if there is an edge from $i$ to $j$, i.e. $(i,j)\in E$, $j$ is called a \emph{neighbor} of $i$. Let $E(i)$ denote all neighbors of agent $i$, i.e., $E(i)=\{j\in N \mid (i,j)\in E \}$.

\paragraph{Delegation Games}
In the delegation game by~\citet{bloembergen19rational} agents' payoffs in an LDE depend solely on the accuracy of their gurus, and the effort agents incur should they vote directly. Here we abstract from the effort element of the model and focus instead on incorporating a power-seeking element in agents' utilities.
The key intuition behind  our extension is to model agents that are not only interested in voting accurately, but also in their own influence during the vote. So our agents choose their delegations by aiming at maximizing the trade-off between pursuing high accuracy and seeking more power.
\begin{definition} \label{def:games}
A delegation game is a tuple $\mathcal{D} = \langle N, G, \{q_i\}_{i\in N}, \Sigma, \beta, u\rangle$, where:
$N$ is a finite set of agents;
$G = \langle N, E \rangle$ is a directed graph;
$q_i$ is $i$'s accuracy;
$\Sigma_i = E(i)\cup\{i\}$ is $i$'s delegation strategy space;
$\beta \in (\frac{n}{2}, n]$ is a quota;
$u$ is the utility function, defined as follows:
\begin{align}
u_i(\d) =  \DB_i(\d) \cdot q_{d^*_i} \label{eq:u}
\end{align}
\end{definition}
Observe that the strategy profiles of this game are delegation profiles. Each such profile $\d$ induces an LDE $\langle N, \omega, \d, \beta \rangle$ where we assume $\omega$ to be the standard weight function assigning weight $1$ to each agent. The utility of profile $\d$ for $i$ is the accuracy that $i$ acquires in $\d$, multiplied by $i$'s power in $\d$, measured by $\DB_i(\d)$.\footnote{Notice that we slightly abuse notation here by using $\d$ directly as input for the index, instead of the corresponding LDE.} Notice that, therefore, the utility of a dummy agent is $0$ and that the utility of a dictator equals her accuracy. 

For our experiments we will be using the more general form of \eqref{eq:u} given by $\DB_i(\d)^\alpha \cdot q_{d^*_i}$, with $\alpha \in [0,1]$. Intuitively, parameter $\alpha$ will be used to control how much agents are influenced by power in the range going from no influence to influence equal to that of accuracy.




\paragraph{Equilibrium Analysis}
In this section, we ask the natural question of whether the games of Definition \ref{def:games} have a Nash equilibrium (NE) in pure strategies.
In general, the answer to this question is negative:
\begin{theorem} \label{th:noNE}
There are delegation games that have no (pure strategy) NE.
\end{theorem}
\begin{proof}[Sketch of proof]
The proof consists in providing a delegation game and then showing that it is possible to construct a profitable deviation for some agent, for every possible delegation profile. 
The game is $N=\{1,2,3,4,5,6\}$, $q_1=0.51, q_2=0.7, q_3=0.9, q_4=0.6, q_5=0.7, q_6=0.9$, $\beta=4$, and the underlying directed graph is $G=\langle N, E=\{(1,3),(2,3),(4,6),(5,6)\}\rangle$.
\end{proof}


However NE can be guaranteed to exist when the underlying network is complete.
\begin{theorem}
\label{thm:NEincomplete}
In any delegation game $\mathcal{D}=\langle N, \{q_i\}_{i\in N}, G, \Sigma, \beta, u\rangle$, where $G$ is a complete network, there exists at least one (pure strategy) NE. 
\end{theorem}
\begin{proof}[Sketch of proof]
First of all observe that on a complete network, if an equilibrium exists, it must be such that all delegation chains are of length $1$ by Fact \ref{fa:forward}.
The theorem is then proven by construction. We construct a profile by letting each agent choose, in turn according to a given sequence, whether to delegate to agent $i^*\in N$, who has the highest accuracy. If an agent decides to delegate to $i^*$, this delegation is assumed to be fixed. So the set of delegators (to $i^*$) monotonically increases. The process continues until no guru wants do delegate. We show that the profile constructed in such a way is a NE by showing that:
\fbox{1} $i^*$ has no profitable deviation because were she deviating to a delegator, she would form a cycle (thereby obtaining utility $0$), and were she delegating to another guru she would inherit a lower accuracy (by assumption) and have lower power (by Fact~\ref{fa:todelegate}).
\fbox{2} No delegator has a profitable deviation, neither by delegating to another delegator (by Fact~\ref{fa:forward}), nor by delegating to another guru (as she would inherit lower accuracy and obtain lower power), nor by becoming herself a guru. The latter claim requires some work and can be shown by proving that as more agents delegate to $i^*$, the power of delegators does not decrease while that of gurus (except for $i^*$) does.
\fbox{3} Finally, no guru has a profitable deviation, neither by delegating to $i^*$ (by construction), nor to another guru (as she would then be better off delegating to $i^*$), nor by delegating to a delegator (again, by Fact~\ref{fa:forward}).
\end{proof}


\section{Experiments}

We use the above model to study, by means of experiments, how the distribution of power in profiles is affected by specific parameters. In particular, we are interested in gaining insights into: 
whether higher connectivity of the underlying network increases power imbalances; whether the more agents are driven by power the less they tend to delegate. 

\medskip

\noindent
{\bf Setup}
Starting from the trivial profile, we generate profiles in two ways: by one-shot interaction (OSI), in which each agent selects their neighbor (including themselves) which maximises their uitility; and by iterated better response dynamics (IBRD), in which each agent iteratively selects one neighbor at random and delegates to her only if this increases her utility, until a stable state (equilibrium) is reached. 

As one might expect, the bottleneck in our experiments consists in the computation of \DB~in order to establish agents' utilities by Eq.~\eqref{eq:u}. It is well-known that computing the Banzhaf index in weighted voting games is intractable \cite{MATSUI2001305}. We therefore implement the approximation method described in \cite{bachrach2008approximating}. Each time we need to establish the $\DB$ of an agent, 15000 coalitions are randomly sampled (by uniform distribution), and the ratio of the coalitions for which the agent is swing is used as the estimator of the $\DB$. By the analytical bounds proven in \cite{bachrach2008approximating}, with the above method we know that the correct $\DB$ is in the confidence interval $[\widehat{\DB}-0.011,\widehat{\DB}+0.011]$ with probability of $0.95$, where $\widehat{\DB}$ is the estimator. So it should be clear that the statistics presented in this section report on values that depend on the estimator $\widehat{\DB}$, and that with high probability are close to the exact intended values.

We will be working with two parameters. To test the effect of connectivity on power we assume that interaction happens on a random network and vary the probability $p$ (see range in e.g., Fig.~\ref{a:ratio}) of any two agents being linked. 
To test the effect of different attitudes towards the importance of power for agents we work with the generalization $\DB_i(\d)^\alpha \cdot q_{d^*_i}$ of \eqref{eq:u} with $\alpha \in \{0, 0.25, 0.5, 0.75, 1 \}$ and assuming an underlying random network with $p = 0.75$. 

We set $|N|=30$, the quota $\beta = 16$, and for each parameter setting, we use an accuracy vector $Q\in \mathbb{R}^{30}$, where each element in $Q$ is drawn from a Gaussian distribution $\mathcal{N}(0.75,0.125)$.
All statistics are the mean value over 50 instances for each parameter setting.
Further details on the setup of our experiments, including pseudo-code for the algorithms of OSI and IRBD are provided in the appendix.
 \begin{figure}[t]
\centering
\subcaptionbox
{A: ratio of delegators\label{a:ratio}}{\includegraphics[width=3.7cm]{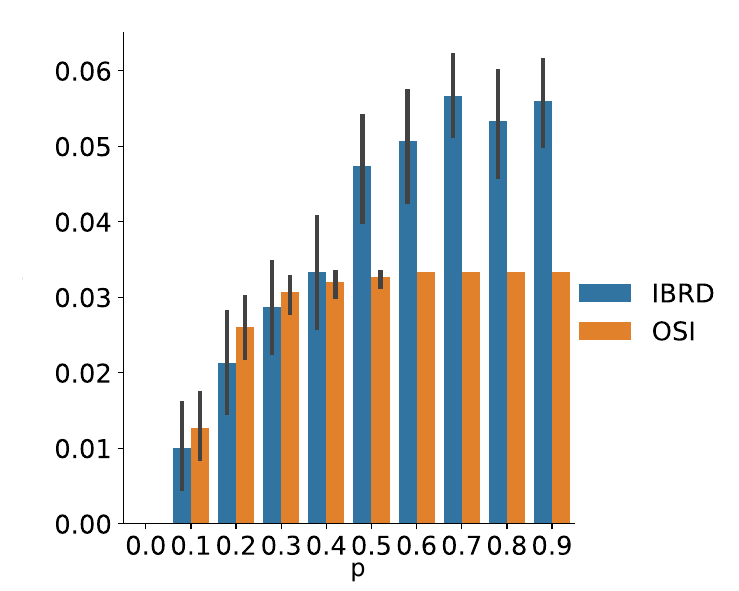}}
\subcaptionbox
{B: ratio of delegators\label{b:ratio}}{\includegraphics[width=3.7cm]{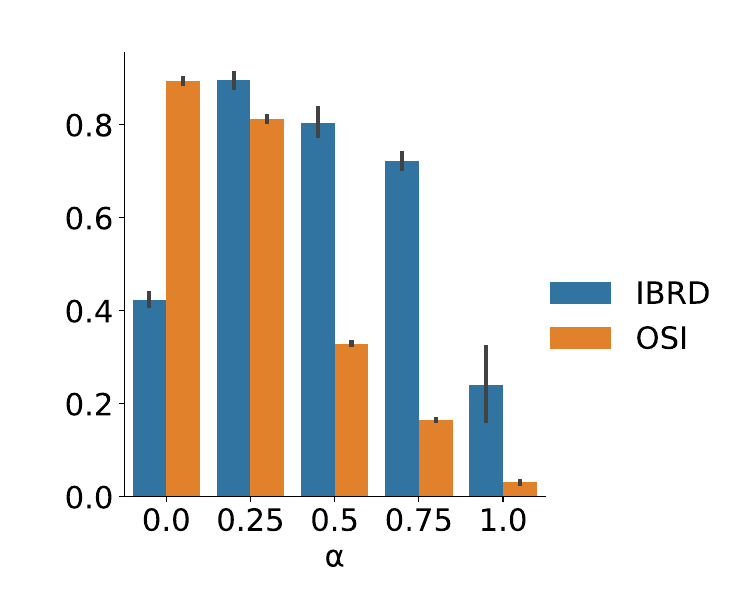}}
\par
\subcaptionbox
{A: average $\DB$\label{a:avgban}}{\includegraphics[width=3.7cm]{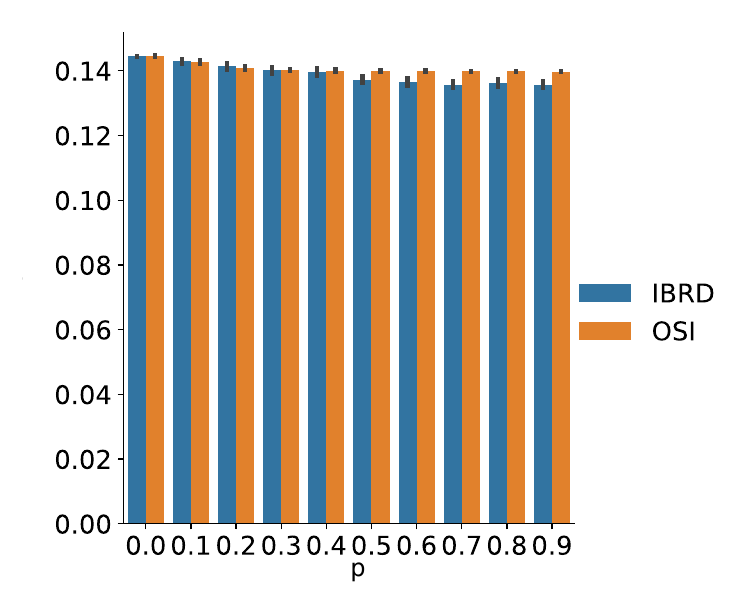}}
\subcaptionbox
{B: average $\DB$\label{b:avgban}}{\includegraphics[width=3.7cm]{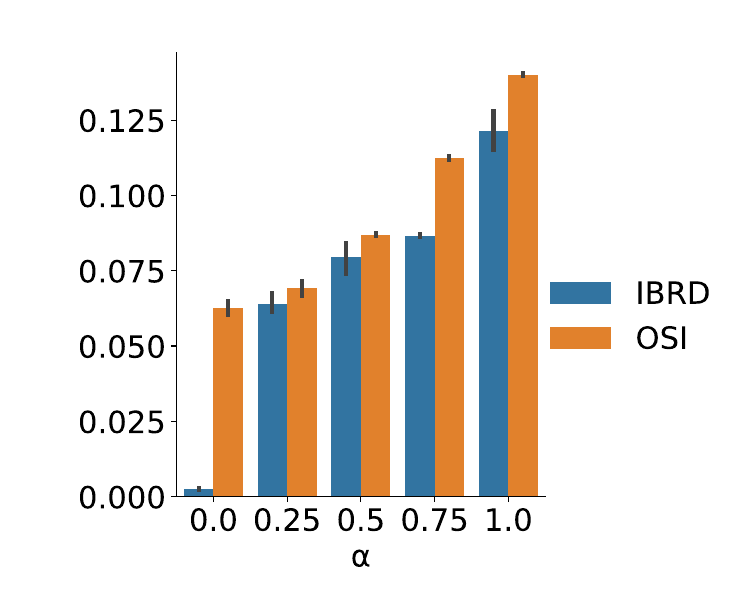}}
\par
\subcaptionbox
{A: Gini coefficient\label{a:gini}}{\includegraphics[width=3.7cm]{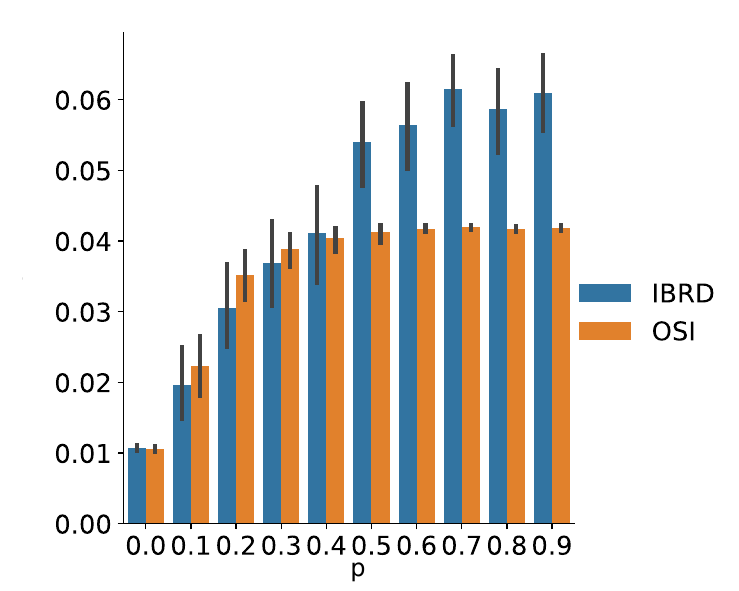}}
\subcaptionbox
{B: Gini coefficient\label{b:gini}}{\includegraphics[width=3.7cm]{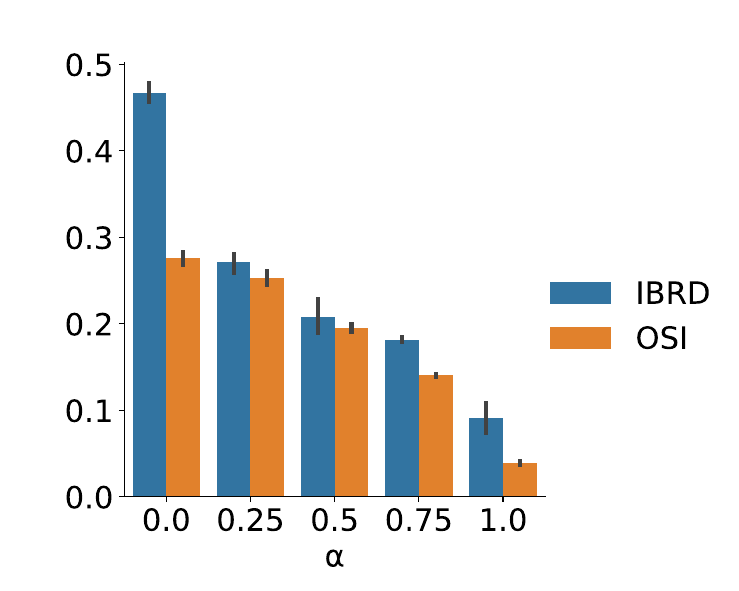}}
 \caption{Selected plots from experiments A and B}
\end{figure}


\medskip

\noindent
{\bf Connectivity: experiment A}
Fig.~\ref{a:ratio} shows that the higher the connectivity (the larger $p$), the more agents tend to delegate both in equilibrium (IBRD) and in one-shot interaction (OSI). This is in line with expectations as agents have more chance to interact with high-accuracy agents. It is worth observing, however, that the ratio of delegators is very low (less than 0.06 on average). That is, very few agents delegate on average. This is in contrast with the behavior of the delegation game where utility is solely based on accuracy (cf. \citet{bloembergen19rational}). We will see with experiment B that the influence of power on agents' utility seems to be an important factor in limiting vs. facilitating delegations.
Despite the small number of delegations we can still observe that increasing $p$ lowers the mean value of \DB (Fig.~\ref{a:avgban}) and increases inequality in the distribution of \DB, measured by the Gini coefficient (Fig.~\ref{a:gini}), although it should be stressed the Gini coefficient remains very low due to the small fraction of delegators. Intuitively, more delegations enhance some agents' power, but reduce the power of other agents, be they gurus or delegators.


\medskip

\noindent
{\bf Power: experiment B}
Fig.~\ref{b:ratio} shows that larger values of $\alpha$ correspond to significantly fewer delegators for OSI. As agents put more weight on power they are more reluctant to delegate in the initial profile (recall Fact \ref{fa:todelegate}). For IRBD this effect is observable only for $\alpha$ in the upper half of the range of available values. We argue this may depend on the fact that IRBD, at the initial profile, allows for delegations to take place that only suboptimally improve utility, triggering then further delegations at later iterations. As $\alpha$ grows, the average power increases (Fig.~\ref{b:avgban}) and inequality in the distribution of power decreases (Fig.~\ref{b:gini}). The decreasing ratio of delegators (Fig. \ref{b:ratio}) is another side of this trend: as $\alpha$ decreases the group turns from consisting mostly of delegators ($\alpha \in \{ 0,  0.25 \}$) to low numbers of delegators comparable to those observed in experiment A ($\alpha = 1$). 

Interestingly, the average length of chains also significantly decreases from around $5.5$ for $\alpha =0$ to $1.2$ for $\alpha =0.25$, and further mildly decreases to around $0.5$ for $\alpha =1$. This is in line with Fact~\ref{fa:forward}: as power becomes more important, agents prefer shorter chains to their gurus.
%

\smallskip

We also ran experiments to identify the effects of varying quota $\beta$ in the range $\{0.6|N|, 0.8|N|, |N|\}$. We identified a similar trend showing that higher quota tend to limit the amount of delegations, but our results were less robust than those reported in experiments A and B.
The refinement of this experiment is left to future work.


\section{Conclusions}

The paper developed a power index for voting with delegable proxy. We showed that the index generalizes the Banzhaf index for standard weighted voting and can be axiomatized in a similar fashion. 
We used the index to model a variant of delegation games for liquid democracy where agents seek to find a tradeoff between increasing their accuracy and acquiring power in the system. We showed equilibria for this sort of interaction exist under a full connectivity condition, but do not exist in general. Finally, two parameters of the model were shown, through simulations, to play an important role in containing the emergence of large inequalities in the distribution of power: the level of connectivity of an underlying (random) network, and the extent to which agents are motivated by the aquisition of power. 

The paper opens several directions for future research of both theoretical and experimental kind. Here we mention one: it would be interesting to 
understand how much agents' attitude towards power could help in readdressing the deterioration of decision-making quality highlighted by \citet{kahng18liquid,caragiannis19contribution}, through its equalizing effect on power distribution.


\paragraph{Acknowledgments}
The authors are grateful to Shuaipeng Liu for helpful insights on the proof of Theorem~\ref{thm:NEincomplete}.


\bibliography{power_bib}{}

\section{Appendix}

\appendix

This technical appendix is structured in two sections:
\begin{itemize}
\item The first section provides full proofs of all results presented in the paper;
\item The second section provides further details on our experimental setup and our simulations.
\end{itemize}


\section{Full Proofs}

\subsection{Characterization of \DB}

We start by fixing some auxiliary notation. Let an LDE $\V=\langle N,\omega, \d, \beta\rangle$ be given.
For any $i\in N$, let $(\d_{-i}, d'_i)$ be the profile in which any $j\in N\ex \{i\}$ delegates as in $\d$, while agent $i$ chooses delegation $d'_i$.
Furthermore, $\d_C$ denotes the profile restricted to a coalition $C\subseteq N$. Such a restricted profile is a mapping $\d_C: C\ra C\cup \{\nul\}$  defined as follows, for all $i\in C$: 
\[
\d_C(i) = 
\left\{
\begin{array}{ll}
\d(i) & \mbox{if} \ \d(i) \in C \\
\nul & \mbox{otherwise}
\end{array}
\right.
\]
That is, in $\d_C$ all agents in $C$ either delegate to agents in $C$ or abstain.
Recall then notation $\hat{\D}(C)$, i.e., the set of all agents that delegate to some guru in $C\subseteq N$ via a delegation chain contained in $C$, in $\d$.
The same set, for a different profile $\d'$, is denoted $\hat{\D}'(C)$.

\smallskip

\begin{proof}[Proof of Lemma \ref{lemma:characterization1}]
To prove the \textbf{MP}, assume that agent $i$ is a dictator in a LDE $V=\langle N,\omega, \d, \beta \rangle$, then for any coalition $C\subseteq N\backslash \{i\}$, $\nu'_V(C)=0$ and $\nu_V(C\cup\{i\})=1$.
Since the number of all possible coalitions $C\subseteq N\ex \{i\}$ is $2^{n-1}$, $\DB_i(V)=1$.

Then \textbf{NP} simply followed by for any $i\in N^{dum}$, $\DB_i(\V)=\sum_{C\subseteq N\setminus \{i\}}(\nu'_\V(C\cup\{i\})-\nu'_\V(C))=0$ by Definition~\ref{df:dummy}.

To show that $\DB$ satisfies the \textbf{SP}, one first has to show that by the way in which weights $\beta_1 \land \beta_2$ and $\beta_1 \vee \beta_2$ are set in Definition \ref{def:completeness}, we have that for any coalition $C\subseteq N_1\cup N_2$, $\nu'_{\V_1\wedge \V_2}(C)=1$ iff $\nu'_{\V_1}(C\cap N_1)=1$ {\em and} $\nu'_{\V_2}(C\cap N_2)=1$, and $\nu'_{\V_1\vee\V_2}(C)=1$ iff $\nu'_{\V_1}(C\cap N_1)=1$ {\em or} $\nu'_{\V_2}(C\cap N_2)=1$. The proof can then proceed with a standard argument.

We first consider any $i\in N_1-N_2$, i.e., agent $i$ is contained in $N_1$ but not in $N_2$.
Let $m_i^V$ denote the number of times that agent $i$ is swing in the delegative simple game for $\V$, i.e., $m_i^V=|\{C\subseteq N\setminus \{i\}\mid \nu'_V(C)=0, \nu'_V(C\cup \{i\})=1\}|$.
Then, if $i$ is swing in $C\subseteq N_1$ in LDE $V_1$, she is also swing in $(C\cup C')\cap N_1$, for any $C'\subseteq N_2-N_1$.
Therefore, in LDE $V_1\vee V_2$, 
$
m_i^{V_1\vee V_2}=m_i^{V_1}2^{|N_2-N_1|},
$ 
where $m_i^{V_1\vee V_2}$ is the number of times that $i$ is swing in LDE $V_1\vee V_2$.
Additionally, since $i\in N_1-N_2$, $m_i^{V_2}=0$, that is, $i$ cannot be swing in LDE $V_2$, which implies $m_i^{\V_1\wedge\V_2}=0$.
Hence we have for $i\in N_1-N_2$, 
$
m_i^{V_1\vee V_2}=m_i^{V_1}2^{|N_2-N_1|}+m_i^{V_2}2^{|N_1-N_2|}-m_i^{V_1\wedge V_2}.
$
Identical equations can be developed for agent $i\in N_2-N_1$ or $i\in N_1\cap N_2$.
We then divide each side of the equation by $2^{|N_1\cup N_2|-1}$ and obtain that, for any $i\in N_1\cup N_2$,
$
\frac{m_i^{V_1\vee V_2}}{2^{|N_1\cup N_2|-1}}=\frac{m_i^{V_1}}{2^{|N_1|-1}}+\frac{m_i^{V_2}}{2^{|N_2|-1}}-\frac{m_i^{V_1\wedge V_2}}{2^{|N_1\cup N_2|-1}},
$
which implies that $\DB_i(V_1\wedge V_2)+\DB_i(V_1\vee V_2)=\DB_i(V_1)+\DB_i(V_2)$.

To prove the \textbf{BP}, we write $\DB_i(V)+\DB_j(V)$ as:
\begingroup\makeatletter\def\f@size{8}\check@mathfonts
\begin{align*}
\begin{split}
&2^{n-1}(\DB_i(V)+\DB_j(V))\\=&\sum_{C\subseteq N\ex \{i\}}(\nu'_V(C\cup\{i\})-\nu'_V(C))+\sum_{C\subseteq N\ex\{j\}}(\nu'_V(C\cup\{j\})-\nu'_V(C))\\
=&\sum_{C\subseteq N\ex\{i,j\}}(\nu'_V(C\cup\{i\})-\nu'_V(C)+\nu'_V(C\cup\{i,j\})-\nu'_V(C\cup\{j\}))\\&+\sum_{C\subseteq N\ex\{i,j\}}(\nu'_V(C\cup\{j\})-\nu'_V(C)+\nu'_V(C\cup\{i,j\})-\nu'_V(C\cup\{i\}))\\
=&2\sum_{C\subseteq N\ex\{i,j\}}(\nu'_V(C\cup\{i,j\})-\nu'_V(C)).
\end{split}
\end{align*}
\endgroup
Since $\DB_{ij}(V')=1/2^{n-2}\sum_{C\subseteq N\setminus \{i,j\}}(\nu'_V(C\cup\{i,j\})-\nu'_V(C))$, $\DB_i(V)+\DB_j(V)=\DB_{ij}(V')$, where $V'$ is the bloc LDE by forming $i$ and $j$ into a bloc.

As for \textbf{ET}, assume that $i$ and $j$ are symmetric agents.
We show that whenever $i$ is a swing agent, so is $j$, and vice versa.
Then $i$ serves as a swing agent in two cases:
\begin{itemize}
\item[(1)] For any $C\subseteq N\ex \{i,j\}$, such that $\nu'_V(C\cup\{i\})-\nu'_V(C)=1$, by Definition~\ref{df:symmetric}, we obtain
$$\nu'_V(C\cup\{j\})-\nu'_V(C)=\nu'_V(C\cup\{i\})-\nu'_V(C)=1.$$
\item[(2)] For any $C\subseteq N\ex \{i,j\}$, such that $\nu'_V(C\cup\{i, j\})-\nu'_V(C\cup\{j\})=1$, by Definition~\ref{df:symmetric}, we obtain that
$$\nu'_V(C\cup\{i, j\})-\nu'_V(C\cup\{i\})=\nu'_V(C\cup\{i, j\})-\nu'_V(C\cup\{j\})=1.$$
\end{itemize}
That is, each time $i$ serves as a swing agent, $j$ also serves as a swing agent once.
By a similar argument, it can be obtained that each time $j$ serves as a swing agent, $i$ also serves as a swing agent once.
Then $\DB_i(V)=\DB_j(V)$.
\end{proof}


\begin{proof}[Proof of Lemma \ref{lemma:characterization2}]
We start by introducing the following claim based on Lemma~\ref{lemma:characterization2}.
\begin{claim}
\label{claim:weakcharacterization}
A power index $f$ for unanimity LDEs satisfies \textbf{MP}, \textbf{NP}, \textbf{SP}, \textbf{ET}, and \textbf{BP}, only if it is $\DB$.
\end{claim}
Then the proof is approached by first showing that, Lemma~\ref{lemma:characterization2} holds if Claim~\ref{claim:weakcharacterization} holds.
Next, we provide the proof that supports Claim~\ref{claim:weakcharacterization}.\\

First, we show that power index $f$ is $\DB$ for any LDE if $f$ is $\DB$ for any unanimity LDE.
Assume that an arbitrary LDE $\V$ is given and let $\mathscr{C}=\{C_1, \dots, C_m\}$ denote all minimally winning coalitions.
Notice that any winning coalition can be represented as the union of a subset of $\mathscr{C}$.
Hence $\V$ can be represented as the disjunction of $m$ unanimity LDEs, i.e., $\V=\V_1\vee \V_2\vee \dots \vee \V_m$ where $\V_j=\langle C_j,\omega, \d_{C_j}, \beta^U\rangle$ ($1\le j\le m$) is a unanimity LDE.
Observe that any agent's delegation strategy is consistent in all unanimity games, that is, the condition of Definition~\ref{def:completeness} is satisfied.

We prove by induction on the size of disjunction of unanimity LDEs.
As the basis, $f_i(\V_j)$ is $\DB$ by the assumption, where $i\in N$ and $1\le j\le m$.
Henceforth, we assume that for any LDE, which is the disjunction of $k$ ($k<m$) unanimity games in $\{\V_1,\dots, \V_m\}$, $f$ is equivalent to $\DB$, then prove that for any LDE, which is the disjunction of $k+1$ unanimity games in $\{\V_1,\dots, \V_m\}$, $f$ is also equivalent to $\DB$.
Without loss of generality, assume that $f_i(\V_1\vee \dots \vee \V_k)$ is $\DB$, and we prove that $f_i(\V_1\vee\dots \vee \V_k\vee \V_{k+1})$ is also $\DB$.
By \textbf{SP}, we have $f_i(\V_1\vee\dots \vee \V_k\vee \V_{k+1})=f_i(\V_1\vee\dots\vee \V_k)+f_i(\V_{k+1})-f_i((\V_1\vee\dots \vee \V_k)\wedge \V_{k+1})$.
Observe that $(\V_1\vee\dots \vee \V_k)\wedge \V_{k+1}=(\V_1\wedge \V_{k+1})\vee\dots\vee (\V_k\wedge \V_{k+1})$.
Since $\V_j$ is a unanimity LDE, $\V_j\wedge \V_{k+1}$ is equivalent to the unanimity LDE $\langle C_j\cup C_{k+1},\omega, \d_{C_j\wedge C_{k+1}}, \beta^U\rangle$.
Therefore, by the assumption that $f$ is $\DB$ for disjunction of $k$ unanimity LDEs, we have $f$ is $\DB$ for $(\V_1\vee\dots \vee \V_k)\wedge \V_{k+1}$.
Hence it implies that $f$ is $\DB$ for $\V_1\vee\dots \vee \V_k\vee \V_{k+1}$.
Intuitively, the number of times that any agent $i\in \bigcup_{1\le j\le k+1}C_j$ serves as a swing agent in $\V_1\vee\dots \vee \V_k\vee \V_{k+1}$ is the sum of her swing times in $\V_1\vee\dots \vee \V_k$ and $\V_{k+1}$, subtracting her swing times in $(\V_1\vee\dots \vee \V_k)\wedge \V_{k+1}$.
Therefore, we proved that if Claim~\ref{claim:weakcharacterization} holds, Lemma~\ref{lemma:characterization2} holds automatically.

\medskip
Next, we prove Claim~\ref{claim:weakcharacterization} by induction on the size of the agent set.
Consider an arbitrary unanimity LDE $\dot{\V}=\langle N,\omega, \d, \beta^U\rangle$.
Let $\hat{N}=N\setminus N^{dum}$ denote all non-dummy agents, and $\hat{n}=|\hat{N}|$.
As the basis of the induction, consider the case in which there is only one agent, i.e., $N=\{i\}$.
If $i$ is a dummy agent, by \textbf{NP}, $f_i(\dot{\V})=0$.
On the other hand, if $i\notin N^{dum}$, $i$ is a dictator, which implies that $f_i(\dot{\V})=1=1/2^{(\hat{n})-1}$ due to \textbf{MP}.

Then we assume that $f$ is \DB\ if $|N|=k$ ($k\in \mathbb{N}_+$, i.e., positive integer), and prove that $f$ is also \DB\ if $|N|=k+1$.
That is in $\dot{\V}$ ($|N|=k+1$), we prove that for any $i\in \hat{N}$, $f_i(\dot{\V})=1/2^{\hat{n}-1}$ and for any $i\in N^{dum}$, $f_i(\dot{\V})=0$, which is identical to \DB.
For any unanimity LDE, let's consider three exhaustive cases: (1) all agents are dummy agents, (2) only one non-dummy agent exists in the unanimity LDE, and (3) more than one non-dummy agents exist.
\begin{case}$N=N^{dum}$.
That is, all agents are dummy agents.
Then by \textbf{NP}, for all $i\in N$, $f_i(\dot{\V})=0$.
\end{case}

\begin{case}$|N^{dum}|=k$.
In this case, there is only one non-dummy agent, denoted by $i$.
Then $i$ is a dictator, and $f_i(\dot{\V})=1=1/2^{\hat{n}-1}$ by \textbf{MP}.
On the other hand, for any $j\in N^{dum}$, $f_i(\dot{\V})=0$ by \textbf{NP}.
\end{case}

\begin{case} $|N|-|N^{dum}|>1$.
When there are more than one non-dummy agents in the unanimity LDE, we further consider three subcases:
\begin{subcase}\label{case3.1}$\forall i\in N\setminus N^{dum}$, $i\in N^{\d}$. 
That is, any non-dummy agent is a guru.
Let $i,j\in N\setminus N^{dum}$, then we form $i$ and $j$ into a bloc and obtain the bloc LDE $\V'$.
Observe that $\V'$ has $k$ agents.
Therefore, by assumption and \textbf{BP}, $f_i(\dot{\V})+f_j(\dot{\V})=f_{ij}(\V')=1/2^{\hat{n}-2}$.
Since $i$ and $j$ are symmetric in $\dot{V}$, $f_i(\V')=f_j(\V')=1/2^{\hat{n}-1}$ due to \textbf{ET}.
Moreover, any agent $i'\in \hat{N}\setminus \{i,j\}$ is symmetric with $i$ (or $j$), thus $f_{i'}(\V')=f_i(\V')=1/2^{\hat{n}-1}$, and for any $a\in N^{dum}$, $f_a(\dot{\V})=0$ by \textbf{NP}.
\end{subcase}
\begin{subcase}\label{case3.2} $N^{\d}=\{i\}$.
In this case, there is only one guru, which is $i$, and any other delegator has $i$ as their guru.
Therefore, for all $j\in N\setminus N^{dum}$, $\d^*(j)=i$.
Assume $j\in N$, such that $d_j=i$.
Then we obtain a bloc game $\V'$ by forming $i$ and $j$ into a bloc $ij$.
By \textbf{BP} and the assumption, $f_i(\dot{\V})+f_j(\dot{\V})=f_{ij}(\V')=1/2^{\hat{n}-2}$.
Then, since \textbf{ET}, $f_i(\dot{\V})=f_j(\dot{\V})=1/2^{\hat{n}-1}$.
Additionally, for all $a\in N^{dum}$, $f_a(\dot{\V})=0$ due to \textbf{NP}, and for all $a\in N\setminus N^{dum}$, $f_a(\dot{\V})=1/2^{\hat{n}-1}$ due to \textbf{ET}.\end{subcase}
\begin{subcase} \label{case3.3}$|N^{\d}|>1$ and $N^{\d}\subset N\setminus N^{dum}$.
In this case there are more than one gurus and at least one delegator delegates to one of the gurus. We can then apply similar arguments to those provided for Case~\ref{case3.1} or Case~\ref{case3.2} to join agents into a bloc and thus prove that $f$ is equivalent to \DB.\end{subcase}
\end{case}

This completes the proof of Claim~\ref{claim:weakcharacterization}.
\end{proof}


\begin{proof}[Proof of Fact \ref{fa:todelegate}]
To prove the fact, it is sufficient to prove that, for any coalition $C\subseteq N$ ($i\in C$), if $i$ is not a swing agent for $C$ in $\V$, neither is she a swing agent for $C$ in $\V'$.
Towards a contradiction, we assume that $i$ is a swing agent for $C$ in $\V'$ even if $i$ is not a swing agent in $\V$.
Then we have that $\sum_{a\in \hat{\D}'(C)}\omega(a)\ge \beta$,
and $\sum_{a\in \hat{\D}'(C\setminus\{i\})}\omega(a)=\sum_{a\in \hat{\D}'(C)}\omega(a)-\omega(i)< \beta$.
Since the only difference between $\d$ and $\d'$ is the strategy of $i$, $\sum_{a\in \hat{\D}(C\setminus\{i\})}\omega(a)=\sum_{a\in \hat{\D}'(C\setminus\{i\})}\omega(a)$, i.e., the weight of $C\setminus \{i\}$ is identical in both LDEs $\V$ and $\V'$.
Moreover, since $i$ is a guru in $\V$ and $i\in C$, it holds that 
$\sum_{a\in \hat{\D}(C)}\omega(a)=\sum_{a\in \hat{\D}(C\setminus\{i\})}\omega(a)+\omega(i)= \sum_{a\in \hat{\D}'(C\setminus\{i\})}\omega(a)+\omega(i)\ge \beta,$
which contradicts the assumption that $i$ is not a swing agent in $\V$.
\end{proof}

\begin{proof}[Proof of Fact \ref{fa:powermonotonicity}]
We prove the fact by showing that $\DB_i(\V)-\DB_j(\V)\le 0$.
By the definition of the delegative simple game of $\V$, we substitute $\DB_i(\V)$ and $\DB_j(\V)$ in $\DB_i(\V)-\DB_j(\V)$ as follows.
\begingroup\makeatletter\def\f@size{8}\check@mathfonts
\begin{align*}
\begin{split}
&2^{n-1}(\DB_j(\V)-\DB_i(\V))\\ = & \sum_{C\subseteq N\ex \{j\}}(\nu'_\V(C\cup\{j\})-\nu'_\V(C))-\sum_{C\subseteq N\ex \{i\}}(\nu'_\V(C\cup \{i\})-\nu'_\V(C))\\
=& \sum_{C\subseteq N\ex \{i,j\}}(\nu'_\V(C\cup\{j\})-\nu'_\V(C)+\nu'_\V(C\cup\{i,j\})-\nu'_\V(C\cup\{i\}))\\
&-\sum_{C\subseteq N\ex \{i,j\}}(\nu'_\V(C\cup\{i\})-\nu'_\V(C)+\nu'_\V(C\cup\{i,j\})-\nu'_\V(C\cup\{j\}))\\
=& 2\sum_{C\subseteq N\ex \{i,j\}}(\nu'_\V(C\cup\{j\})-\nu'_\V(C\cup\{i\})).
\end{split}
\end{align*}\endgroup
Concerning the above equation, for any $C\subseteq N\setminus \{i,j\}$, we consider two possible cases:
\paragraph{(1) $\nu'_V(C\cup\{j\})=0$.}
It implies that $\sum_{a\in \hat{\D}(C\cup\{j\})}\omega(a)<\beta$, and consequently, $\sum_{a\in \hat{\D}(C)}\omega(a)\le \sum_{a\in \hat{\D}(C\cup\{j\})}\omega(a)<\beta$.
Since $d_i=j$ and $j\notin C$, $i\notin \hat{\D}(C\cup\{i\})$.
Therefore, we have $\sum_{a\in \hat{\D}(C\cup\{i\})}\omega(a)=\sum_{a\in \hat{\D}(C)}\omega(a)<\beta$, which implies that $\nu'_\V(C\cup\{i\})=0$.
Hence $\nu'_\V(C\cup\{j\})-\nu'_\V(C\cup\{i\})=0$.

\paragraph{(2) $\nu_V(C\cup\{j\})=1$.}
It implies that $\sum_{a\in \hat{\D}(C\cup\{j\})}\omega(a)\ge \beta$.
Then we consider two possible cases:\\
(i). $\sum_{a\in \hat{\D}(C)}\omega(a)<\beta$.
Since $i\notin \hat{\D}(C\cup\{i\})$, it can be inferred that $\sum_{a\in \hat{\D}(C\cup\{i\})}\omega(a)=\sum_{a\in \hat{\D}(C)}\omega(a)<\beta$, which implies that $\nu'_\V(C\cup\{i\})=0$.
Therefore, $\nu'_\V(C\cup\{j\})-\nu'_\V(C\cup\{i\})>0$.\\
(ii). $\sum_{a\in \hat{\D}(C)}\omega(a)\ge \beta$.
We can obtain that $\sum_{a\in \hat{\D}(C\cup\{i\})}\omega(a)\ge\sum_{a\in \hat{\D}(C)}\omega(a)\ge \beta$, which implies that $\nu'_\V(C\cup\{i\})=1$.
Therefore, $\nu'_\V(C\cup\{j\})-\nu'_\V(C\cup\{i\})=0$.\\
Hence, to sum up, we have$$\sum_{C\subseteq N\ex \{i,j\}}(\nu'_\V(C\cup\{j\})-\nu'_\V(C\cup\{i\}))\ge 0,$$
which implies that 
$$\DB_j(\V)-\DB_i(\V)\ge 0.$$
\end{proof}

\begin{proof}[Proof of Fact \ref{fa:forward}]
It is sufficient to show that if $k$ is a swing agent for coalition $C$ in LDE $\V$, she is also a swing agent for $C$ in $\V'$.
Then we have $\sum_{a\in\hat{\D}(C)}\omega(a)\ge \beta$ and $\sum_{a\in \hat{\D}(C\ex \{k\})}\omega(a)<\beta$.
It implies that $k\in \hat{\D}(C)$, from which we can infer that $i,j\in \hat{\D}(C)$ since $i$ and $j$ are among intermediaries between $k$ and $\d^*(k)$.
Note that the only difference between $\d$ and $\d'$ is the strategy of $k$, and $\d(k)=i$ while $\d'(k)=j$.
Therefore, we obtain $k\in \hat{\D}'(C)$ since all intermediaries between $k$ and $\d'^*(k)$ are contained in $C$, and consequently $\hat{\D}(C)=\hat{\D}'(C)$.
Hence $\sum_{a\in\hat{\D}'(C)}\omega(a)=\sum_{a\in\hat{\D}(C)}\omega(a)\ge \beta$ and $\sum_{a\in \hat{\D}'(C\ex \{k\})}\omega(a)=\sum_{a\in \hat{\D}(C\ex \{k\})}\omega(a)<\beta$, which implies that $k$ is also a swing agent for $C$ in $\V'$.
\end{proof}


\subsection{(In)Existence of Nash equilibria}


\begin{proof}[Proof of Theorem \ref{th:noNE}]
Consider the delegation game defined as follows.
$N=\{1,2,3,4,5,6\}$, $q_1=0.51, q_2=0.7, q_3=0.9, q_4=0.6, q_5=0.7, q_6=0.9$, $\beta=4$, and for the underlying graph $R=\langle N, E\rangle$, $E=\{(1,3),(2,3),(4,6),(5,6)\}$, which can be represented as Fig.~\ref{fig:counterNE}.
Notice that this is a directed graph.
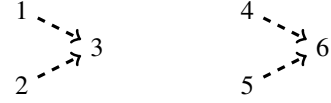
\begin{figure}[h!]
\centering
\begin{tikzpicture}
\path (0,1) node(1){1} (0,0) node(2){2} (1,0.5) node(3){3} (3,1) node(4){4} (3,0) node(5){5} (4,0.5) node(6){6};
\draw[dashed,very thick,->](1)--(3);
\draw[dashed,very thick,->](2)--(3);
\draw[dashed,very thick,->](4)--(6);
\draw[dashed,very thick,->](5)--(6);
\end{tikzpicture}
\caption{The underlying directed graph.}
\label{fig:counterNE}
\end{figure}
\\
We show that in each possible profile, there exists an agent who has incentive to deviate.
As follows, each possible profile with a corresponding deviating agent is listed.
\begin{itemize}
\item The trivial profile $\d^0$, in which each agent is a guru, as shown in Fig~\ref{fig:nodelegate}.\\
(1) $\d^0$. Agent $1$ deviates from $\d^0$ to $\d^1$, $\DB_1(\d^0)=3/16$ to $\DB_1(\d^1)=5/16$, and $u_1(\d^0)=0.1594$ to $u_1(\d^1)=0.1688$.
\begin{figure}[h!]
\centering
\begin{tikzpicture}
\path (0,1) node(1){1} (0,0) node(2){2} (1,0.5) node(3){3} (2,1) node(4){4} (2,0) node(5){5} (3,0.5) node(6){6};
\end{tikzpicture}
\caption{The trivial profile $\d^1$.}
\label{fig:nodelegate}
\end{figure}
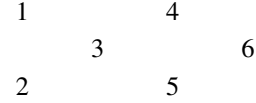
\item Profiles with only one delegating agent, as shown in Fig~\ref{fig:onedelegate}.\\
(2) $d^1$. Agent $4$ deviates from $\d^1$ to $\d^7$, $\DB_1(\d^1)=1/4$ to $\DB_1(\d^1)=3/16$, and $u_1(\d^0)=0.15$ to $u_1(\d^1)=0.1688$.\\
(3) $\d^2$. Agent $1$ deviates from $\d^2$ to $\d^5$, $\DB_1(\d^2)=1/4$ to $\DB_1(\d^5)=3/16$, and $u_1(\d^0)=0.1275$ to $u_1(\d^1)=0.1688$.\\
(4) $\d^3$. Agent $1$ deviates from $\d^3$ to $\d^7$, $\DB_1(\d^3)=1/4$ to $\DB_1(\d^7)=5/32$, and $u_1(\d^3)=0.1275$ to $u_1(\d^7)=0.1406$.\\
(5) $\d^4$. Agent $1$ deviates from $\d^4$ to $\d^8$, $\DB_1(\d^4)=1/4$ to $\DB_1(\d^8)=5/32$, and $u_1(\d^4)=0.1275$ to $u_1(\d^8)=0.1406$.
\begin{figure}[h!]
\centering
\hfill
\subcaptionbox
{$\d^1$}{\begin{tikzpicture}[scale=0.9]
\path (0,1) node(1){1} (0,0) node(2){2} (1,0.5) node(3){3} (2,1) node(4){4} (2,0) node(5){5} (3,0.5) node(6){6};
\draw[very thick,->](1)--(3);
\end{tikzpicture}}
\label{fig:d1}
\hfill
\subcaptionbox
{$\d^2$}{\begin{tikzpicture}[scale=0.9]
\path (0,1) node(1){1} (0,0) node(2){2} (1,0.5) node(3){3} (2,1) node(4){4} (2,0) node(5){5} (3,0.5) node(6){6};
\draw[very thick,->](2)--(3);
\end{tikzpicture}}
\label{fig:d2}
\hfill
\subcaptionbox
{$\d^3$}{\begin{tikzpicture}[scale=0.9]
\path (0,1) node(1){1} (0,0) node(2){2} (1,0.5) node(3){3} (2,1) node(4){4} (2,0) node(5){5} (3,0.5) node(6){6};
\draw[very thick,->](4)--(6);
\end{tikzpicture}}
\label{fig:d3}
\hfill
\subcaptionbox
{$\d^4$}{\begin{tikzpicture}[scale=0.9]
\path (0,1) node(1){1} (0,0) node(2){2} (1,0.5) node(3){3} (2,1) node(4){4} (2,0) node(5){5} (3,0.5) node(6){6};
\draw[very thick,->](5)--(6);
\end{tikzpicture}}
\label{fig:d4}
\caption{Profiles with only one delegating agent.}
\label{fig:onedelegate}
\end{figure}
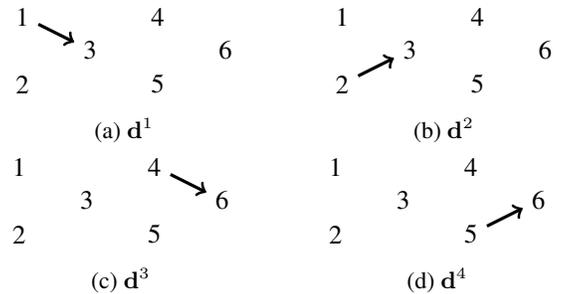

\item Profiles with two delegating agents, as shown in Fig~\ref{fig:twodelegate}.\\
(6) $\d^5$. Agent $2$ deviates from $\d^5$ to $\d^1$, $\DB_2(\d^5)=3/16$ to $\DB_2(\d^1)=1/4$, and $u_2(\d^5)=0.1688$ to $u_2(\d^1)=0.175$.\\
(7) $\d^6$. Agent $5$ deviates from $\d^6$ to $\d^3$, $\DB_5(\d^6)=3/16$ to $\DB_5(\d^3)=1/4$, and $u_2(\d^5)=0.1688$ to $u_2(\d^1)=0.175$.\\
(8) $\d^7$. Agent $4$ deviates from $\d^7$ to $\d^1$, $\DB_4(\d^7)=5/32$ to $\DB_4(\d^1)=1/4$, and $u_4(\d^7)=0.1406$ to $u_4(\d^1)=0.15$.\\
(9) $\d^8$. Agent $5$ deviates from $\d^8$ to $\d^1$, $\DB_5(\d^8)=5/32$ to $\DB_5(\d^1)=1/4$, and $u_5(\d^8)=0.1406$ to $u_5(\d^1)=0.175$.\\
(10) $\d^9$. Agent $4$ deviates from $\d^9$ to $\d^2$, $\DB_4(\d^9)=5/32$ to $\DB_4(\d^2)=1/4$, and $u_4(\d^9)=0.1406$ to $u_4(\d^2)=0.15$.\\
(11) $\d^{10}$. Agent $5$ deviates from $\d^{10}$ to $\d^2$, $\DB_5(\d^{10})=5/32$ to $\DB_5(\d^2)=1/4$, and $u_5(\d^{10})=0.1406$ to $u_5(\d^2)=0.175$.
\begin{figure}[h!]
\centering
\hfill
\subcaptionbox
{$\d^5$}{\begin{tikzpicture}[scale=0.9]
\path (0,1) node(1){1} (0,0) node(2){2} (1,0.5) node(3){3} (2,1) node(4){4} (2,0) node(5){5} (3,0.5) node(6){6};
\draw[very thick,->](1)--(3);
\draw[very thick,->](2)--(3);
\end{tikzpicture}}
\label{fig:d5}
\hfill
\subcaptionbox
{$\d^6$}{\begin{tikzpicture}[scale=0.9]
\path (0,1) node(1){1} (0,0) node(2){2} (1,0.5) node(3){3} (2,1) node(4){4} (2,0) node(5){5} (3,0.5) node(6){6};
\draw[very thick,->](4)--(6);
\draw[very thick,->](5)--(6);
\end{tikzpicture}}
\label{fig:d6}
\hfill
\subcaptionbox
{$\d^7$}{\begin{tikzpicture}[scale=0.9]
\path (0,1) node(1){1} (0,0) node(2){2} (1,0.5) node(3){3} (2,1) node(4){4} (2,0) node(5){5} (3,0.5) node(6){6};
\draw[very thick,->](1)--(3);
\draw[very thick,->](4)--(6);
\end{tikzpicture}}
\label{fig:d7}
\hfill
\subcaptionbox
{$\d^8$}{\begin{tikzpicture}[scale=0.9]
\path (0,1) node(1){1} (0,0) node(2){2} (1,0.5) node(3){3} (2,1) node(4){4} (2,0) node(5){5} (3,0.5) node(6){6};
\draw[very thick,->](1)--(3);
\draw[very thick,->](5)--(6);
\end{tikzpicture}}
\label{fig:d8}
\hfill
\subcaptionbox
{$\d^9$}{\begin{tikzpicture}[scale=0.9]
\path (0,1) node(1){1} (0,0) node(2){2} (1,0.5) node(3){3} (2,1) node(4){4} (2,0) node(5){5} (3,0.5) node(6){6};
\draw[very thick,->](2)--(3);
\draw[very thick,->](4)--(6);
\end{tikzpicture}}
\label{fig:d9}
\hfill
\subcaptionbox
{$\d^{10}$}{\begin{tikzpicture}[scale=0.9]
\path (0,1) node(1){1} (0,0) node(2){2} (1,0.5) node(3){3} (2,1) node(4){4} (2,0) node(5){5} (3,0.5) node(6){6};
\draw[very thick,->](2)--(3);
\draw[very thick,->](5)--(6);
\end{tikzpicture}}
\label{fig:d10}
\caption{Profiles with two delegating agents.}
\label{fig:twodelegate}
\end{figure}
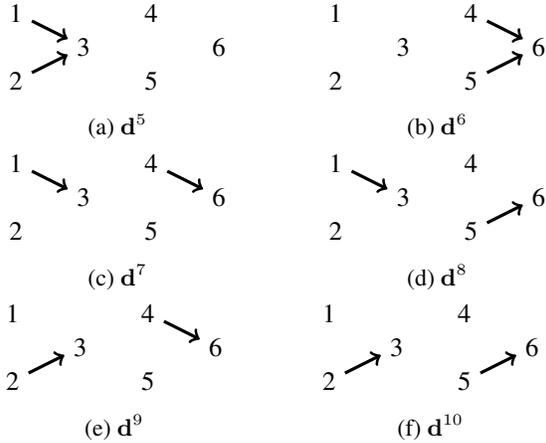
\item Profiles with three delegating agents, as shown in Fig~\ref{fig:threedelegate}.
(12) $\d^{11}$. Agent $2$ deviates from $\d^{11}$ to $\d^6$, $\DB_2(\d^{11})=3/32$ to $\DB_2(\d^6)=3/16$, and $u_2(\d^{11})=0.0844$ to $u_2(\d^6)=0.1313$.\\
(13) $\d^{12}$. Agent $1$ deviates from $\d^{12}$ to $\d^6$, $\DB_1(\d^{12})=3/32$ to $\DB_1(\d^6)=3/16$, and $u_1(\d^{12})=0.0844$ to $u_1(\d^6)=0.0956$.\\
(14) $\d^{13}$. Agent $5$ deviates from $\d^{13}$ to $\d^5$, $\DB_5(\d^{13})=3/32$ to $\DB_5(\d^5)=3/16$, and $u_5(\d^{13})=0.0844$ to $u_5(\d^5)=0.1313$.\\
(15) $\d^{14}$. Agent $4$ deviates from $\d^{14}$ to $\d^5$, $\DB_4(\d^{14})=3/32$ to $\DB_4(\d^5)=3/16$, and $u_4(\d^{14})=0.0844$ to $u_4(\d^5)=0.1125$.
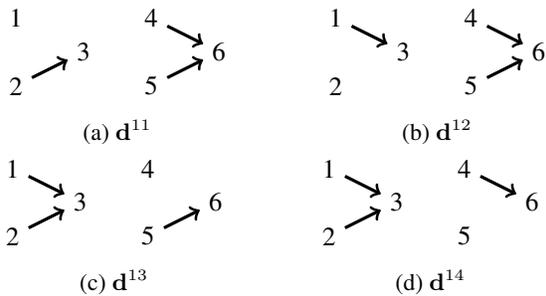
\begin{figure}[h!]
\centering
\hfill
\subcaptionbox
{$\d^{11}$}{\begin{tikzpicture}[scale=0.9]
\path (0,1) node(1){1} (0,0) node(2){2} (1,0.5) node(3){3} (2,1) node(4){4} (2,0) node(5){5} (3,0.5) node(6){6};
\draw[very thick,->](2)--(3);
\draw[very thick,->](4)--(6);
\draw[very thick,->](5)--(6);
\end{tikzpicture}}
\label{fig:d11}
\hfill
\subcaptionbox
{$\d^{12}$}{\begin{tikzpicture}[scale=0.9]
\path (0,1) node(1){1} (0,0) node(2){2} (1,0.5) node(3){3} (2,1) node(4){4} (2,0) node(5){5} (3,0.5) node(6){6};
\draw[very thick,->](1)--(3);
\draw[very thick,->](4)--(6);
\draw[very thick,->](5)--(6);
\end{tikzpicture}}
\label{fig:d12}
\hfill
\subcaptionbox
{$\d^{13}$}{\begin{tikzpicture}[scale=0.9]
\path (0,1) node(1){1} (0,0) node(2){2} (1,0.5) node(3){3} (2,1) node(4){4} (2,0) node(5){5} (3,0.5) node(6){6};
\draw[very thick,->](1)--(3);
\draw[very thick,->](2)--(3);
\draw[very thick,->](5)--(6);
\end{tikzpicture}}
\label{fig:d13}
\hfill
\subcaptionbox
{$\d^{14}$}{\begin{tikzpicture}[scale=0.9]
\path (0,1) node(1){1} (0,0) node(2){2} (1,0.5) node(3){3} (2,1) node(4){4} (2,0) node(5){5} (3,0.5) node(6){6};
\draw[very thick,->](1)--(3);
\draw[very thick,->](2)--(3);
\draw[very thick,->](4)--(6);
\end{tikzpicture}}
\label{fig:d14}
\caption{Profiles with three delegating agents.}
\label{fig:threedelegate}
\end{figure}
\item Due to the restriction by the underlying graph, agent $3$ and agent $6$ can only be gurus.
Therefore in the last possible profile $\d^15$, agents $1$, $2$, $4$ and $5$ are all delegating agents (Fig~\ref{fig:alldelegate}).\\
(16) $\d^{15}$. Agent $5$ deviates from $\d^{15}$ to $\d^{14}$, $\DB_5(\d^{15})=3/32$ to $\DB_5(\d^{14})=5/32$, and $u_5(\d^{15})=0.0844$ to $u_5(\d^{14})=0.1094$.
\begin{figure}[h!]
\centering
\begin{tikzpicture}
\path (0,1) node(1){1} (0,0) node(2){2} (1,0.5) node(3){3} (2,1) node(4){4} (2,0) node(5){5} (3,0.5) node(6){6};
\draw[very thick,->](1)--(3);
\draw[very thick,->](2)--(3);
\draw[very thick,->](4)--(6);
\draw[very thick,->](5)--(6);
\end{tikzpicture}
\caption{Profile $\d^{15}$.}
\label{fig:alldelegate}
\end{figure}
\end{itemize}
Therefore, there is no NE in this delegation game.
\end{proof}


\begin{proof}[Proof of Theorem \ref{thm:NEincomplete}]
Note that in the complete network $G$, any agent can observe and interact with any other agents. 
Hence by Fact~\ref{fa:forward}, no delegation chain is longer than 1.
Then, we prove by construction, that is, we use Algorithm~\ref{algo:NEincomplete} to output a profile and verify that the profile is a NE.\\
To introduce Algorithm~\ref{algo:NEincomplete}, we first introduce a sequence $\sigma$ over $N\setminus \{i^*\}$, where $i^*$ is the agent with the highest accuracy (ties are broken lexicographically).
The sequence is a bijection $\sigma: N\setminus \{i^*\}\longleftrightarrow [n-1]$ ($[k]=\{1,\dots, k\}$ for any $k\in \mathbb{N}_+$).
Let $\sigma(k)$ denote the $k$-th agent in the sequence, where $k\in [n-1]$.
Additionally, for any coalition $C\subseteq N$, let $\sigma_C$ denote the sequence, which is consistent with $\sigma$ but restricted to agents in $C$.\\
\begin{algorithm}[t]
\caption{Nash Equilibrium Construction}
\label{algo:NEincomplete}
\begin{description}
\item[Initialization:]
$i^*$, $C^0=\emptyset$, $C^1=N\ex \{i^*\}$, $\sigma$, $j=1$, $\d:$ for any $i\in N$, $\d(i)=i$.
\item[Delegation:]
\item
\begin{algorithmic}
\WHILE{$C^j\not= C^{j-1}$}
\STATE{$C^{j+1}\la C^j$}
\FOR{$k=1$ to $|C^j|$}
\STATE{$\d(i)\la arg\max_{a\in\{i^*,\sigma_{C^j}(k)\}}u_i((\d_{-\sigma_{C^j}(k)},d'_{\sigma_{C^j}(k)}=a))$}
\IF{$\d(\sigma_{C^j}(k))=i^*$}
\STATE{$C^{j+1}\la C^{j+1}\ex \{\sigma_{C^j}(k)\}$}
\ENDIF
\ENDFOR
\STATE{$j\la j+1$}
\ENDWHILE
\end{algorithmic}
\item[Return:] $\d$
\end{description}
\end{algorithm}
In other words, in Algorithm~\ref{algo:NEincomplete}, in turns determined by $\sigma$, each agent in $N\setminus \{i^*\}$ chooses between being a guru or delegating to $i^*$.
If an agent changes from being a guru to delegating to $i^*$ (to obtain higher utility), she cannot change her strategy anymore.
When no agent wants to change, the algorithm terminates and returns the profile $\d$.\\
Then we verify that $\d$ is a NE.
We show that in $\d$, (1) $i^*$ will not deviate, (2) any delegator will not deviate, and (3) any guru, except for $i^*$, will not deviate.
First, (1) obviously holds since (i) $i^*$ will not change to delegate to any delegator to form a delegating cycle;
(ii) by Fact~\ref{fa:todelegate}, $i^*$ will not delegate to any other guru, otherwise she obtains a lower $\DB$ and inherits a lower accuracy.\\
Next we show (2).
It is clear that a delegator would not change to delegate to another delegator by Fact~\ref{fa:forward}.
Then we show that any delegator will not deviate to be a guru.
We use Lemma~\ref{prop:ne1}, which illustrates if more agents delegate to $i^*$ all current delegators' $\DB$ would not change, and Lemma~\ref{prop:ne2}, which illustrates if more agents delegate to $i^*$, all remained gurus' (except for $i^*$) $\DB$ will be weakly worse off.
\begin{lemma}
\label{prop:ne1}
Given a delegation game $\mathcal{D}$ and a profile $\d$, such that for any $j\in N\setminus N^\d$, $\d(j)=i^*$, we construct another profile $\d'=(\d_{-i},\d'(i)=i^*)$, where $\d(i)=i$ and $i\not=i^*$.
Then we have for all $j\in N$, such that $\d(j)=i^*$, $\DB_j(\d)=\DB_j(\d')$.
\end{lemma}
\begin{proof}
We show that for any delegator $j$ under $\d$, she is a swing agent for any coalition $C\subseteq N$ under $\d$ if and only if she is a swing agent for $C$ under $\d'$.
Since the only difference between $\d$ and $\d'$ is the delegation strategy of $i$, and $i$ is a guru under $\d$ and $\d'(i)=i^*$, for any coalition $C\subseteq N$, such that $i^*\in C$, we have $|\hat{\D}(C)|=|\hat{\D}'(C)|$.
Then note that $j$ can be a swing agent only if she is contained in $\hat{\D}(C)$ under $\d$ or in $\hat{\D}'(C)$ under $\d'$.
Since $\d(j)=\d'(j)=i^*$, we have that if $j$ is contained in $\hat{\D}(C)$ (resp. $\hat{\D}'(C)$), $i^*\in C$ must hold under $\d$ (resp. under $\d'$).
Therefore, we have that under $\d$, $|\hat{\D}(C)|\ge\beta$ and $|\hat{\D}(C\setminus\{j\})|<\beta$ if and only if $|\hat{\D}'(C)|\ge\beta$ and $|\hat{\D}'(C\setminus\{j\})|<\beta$.
Thus $\DB_j(\d)=\DB_j(\d')$.
\end{proof}
\begin{lemma}
\label{prop:ne2}
Given a delegation game $\mathcal{D}$ and a profile $\d$, such that for any $j\in N\setminus N^{\d}$, $\d(j)=i^*$, let $\d'=(\d_{-i},\d'(i)=i^*)$, where $i\in N^{\d}\setminus \{i^*\}$.
Then we have for all $j\in N\setminus N^{\d'}$, $\DB_j(\d)\ge \DB_j(\d')$.
\end{lemma}
\begin{proof}
We compare the times of any $j\in N\setminus N^{\d'}$ serving as a swing agent under $\d$ and $\d'$.
Since the only difference between $\d$ and $\d'$ is the strategy of agent $i$, it is sufficient to consider coalitions containing $i$.
Then, under $\d$ and $\d'$ respectively, we count the number of coalitions for which $j$ is a swing agent.
Consider two possible cases: (1) $i^*\in C$ and (2) $i^*\notin C$.\\
(1). Since $i$ is a guru under $\d$ and $\d'(i)=i^*$, $\hat{\D}(C)=\hat{\D}'(C)$.
Therefore, $|\hat{\D}(C)|\ge \beta$ and $|\hat{\D}(C\setminus\{j\})|<\beta$, if and only if $|\hat{\D}'(C)|\ge \beta$ and $|\hat{\D}'(C\setminus\{j\})|<\beta$.
That is, $j$ serves as a swing agent for $C$ under $\d$ if and only if $j$ is also a swing agent for $C$ under $\d'$.\\
(2). By $i^*\notin C$ and $\d'(i)=i^*$, $|\hat{\D}(C)|=|\hat{\D}'(C)|+1$, since $i\in \hat{\D}(C)$ while $i\notin \hat{\D}'(C)$ due to the lacking of $i^*$ in $C$.
Then we consider two possible (exhaustive) sub-cases:\\
(i). $j$ is a swing agent for $C$ under $\d$, but becomes a non-swing agent for $C$ under $\d'$.
That is, $|\hat{\D}(C)|=\lceil \beta \rceil$ and $|\hat{\D}(C\ex \{j\})|=|\hat{\D}(C)|-1=|\hat{\D}'(C)|=\lceil\beta\rceil-1$ by the fact that $j$ is a guru in both profiles.
Since $i^*\notin C$, none of delegators is contained in $\hat{\D}(C)$ or $\hat{\D}'(C)$.
Then let $C^*$ denote the set of gurus, except for $i^*$, under $\d$, i.e., $C^*=N^{\d}\ex \{i^*\}$, and $n^*=|C^*|$.
Therefore, $|\hat{\D}(C)|=|\hat{\D}(C\cap C^*)|$.
Thus in this sub-case, the number of such coalition $C$, for which $j$ is a swing agent, is ${n^*-2\choose \lceil\beta\rceil-2}$, i.e., $C$ contains $\lceil \beta \rceil-2$ agents in $N^{\d}\ex \{i,j,i^*\}$, and $\{i,j\}$.\\
(ii). $j$ is a swing agent for $C$ under $\d'$, but is a non-swing agent for $C$ under $\d$.
That is $|\hat{\D}(C)|=\lceil \beta \rceil+1$ and $|\hat{\D}'(C)|=|\hat{\D}(C)|-1=\lceil \beta\rceil$.
Then, the number of $C$, for which $j$ is a swing agent under $\d'$, is ${n^*-2\choose \lceil\beta \rceil-1}$, i.e., $C$ contains $\lceil \beta \rceil -1$ agents in $N^{\d}\ex \{i,j,i^*\}$, and $\{i,j\}$.

\smallskip
Since $\beta\ge \lceil \frac{n}{2}\rceil$, we have ${n^*-2\choose \lceil\beta\rceil-2}\ge {n^*-2\choose \lceil\beta \rceil-1}$ since $n^*\le n-1$.
Therefore, the number of times that $j$ serves as a swing agent under $\d$ is weakly more than that under $\d'$.
\end{proof}
Therefore, in Algorithm~\ref{algo:NEincomplete}, if an agent $i$ chooses to delegate to $i^*$, she has no incentive to change back to be a guru under $\d$ since:\\
(1) by Lemma~\ref{prop:ne1}, as more agents delegate to $i^*$, $i$'s utility does not change since her $\DB$ and $q_{i^*}$ do not change;
(2) by Lemma~\ref{prop:ne2}, if she deviates to be a guru, her utility becomes even lower than that before she chooses to delegate to $i^*$.
Next we show that any delegator will not change to delegate to another guru, by using the following Lemma.
\begin{lemma}

\label{prop:ne3}
Given a delegation game $\mathcal{D}$ and a profile $\d$, such that for any $j\in N\ex N^{\d}$, $\d(j)=i^*$, let $\d'=(\d_{-i},\d'(i)=i')$, where $\d(i)=i^*$ and $i'\in N^{\d}\ex \{i^*\}$.
Then $\DB_i(\d)\ge \DB_i(\d')$.
\end{lemma}
\begin{proof}
We also prove the lemma by comparing the number of times that $i$ serves as a swing agent under $\d$ and $\d'$, respectively.
First notice that for any coalition $C$ such that $i^*,i'\in C$, $|\hat{\D}(C)|=|\hat{\D}'(C)|$.
That is, $i$ is a swing agent for $C$ under $\d$ if and only if she is also a swing agent for $C$ under $\d'$.\\
Then we consider any coalition $C$, which contains only one of $i^*$ and $i'$.
Let $C'=N\ex (\D(i^*)\cup\{i'\})$, i.e., all agents except for all delegators under $\d$, $i^*$ and $i'$, and let $n'=|C'|$ and $n^*=|\D(i^*)\ex \{i^*\}|$ (all delegators under $\d$).
Then we consider two cases:\\
(1) $i$ is swing for $C$ under $\d$, but not swing under $\d'$.
We can infer that $i^*\in C$ but $i'\notin C$, since $\d(i)=i^*$ and then $i\in \hat{\D}(C)$.
Therefore, $|\hat{\D}(C)|=\lceil \beta \rceil$ and $i,i^*\in C$, and the number of such coalitions (or times $i$ is swing for $C$ in this case) is ${n'+n^*\choose \lceil\beta \rceil-2}$, i.e., $C$ consists of $\lceil\beta \rceil-2$ agents in $C'\cup(\D(i^*)\ex \{i^*\})$, and $\{i,i^*\}$.\\
(2) $i$ is not swing for $C$ under $\d$, but is swing under $\d'$.
Then we have that $i'\in C$, but $i^*\notin C$.
Therefore, $|\hat{\D}'(C)|=\lceil \beta \rceil$ and $i,i'\in C$,
and the number of such coalitions (or the times of $i$ being a swing agent in this case) is ${n'\choose \lceil\beta \rceil-2}*2^{n^*}$.
That is, $C$ contains $\lceil\beta \rceil-2$ agents in $C'$ as well as $\{i,i'\}$, and since $i^*$ is not in $C$, the emergence of any agent delegating to $i^*$ does not influence the value of $|\hat{\D}(C)|$, thus it leads to $2^{n^*}$ times of ${n'\choose \lceil\beta \rceil-2}$.

\smallskip
Since $\beta\ge \lceil \frac{n}{2}\rceil$, we have ${n'+n^*\choose \lceil\beta \rceil-2}\ge {n'\choose \lceil\beta \rceil-2}*2^{n^*}$ (by Lemma~\ref{lemma6}), thus $\DB_i(\d)\ge \DB_i(\d')$.
\end{proof}


\begin{lemma}
\label{lemma6}
Given $n,n',n^*\in \mathbb{Z}_+$ and $\beta\in (n/2,n]$, such that $n=n^*+n'+2$, we have ${n'+n^*\choose \cb-2}\ge {n'\choose \cb-2}*2^{n^*}$.
\end{lemma}
\begin{proof}
Let $\alpha=\frac{(\cb-2)!(n'-\cb+2)!(n^*+n'-\cb+2)!}{n'!}$.
We assume that $n'-\cb+2\ge 0$, otherwise it obviously holds.
Then 
{\small
\begin{align}
\label{eq:ne31}
\begin{split}
&{n'\choose \cb-2}*2^{n^*}*\alpha\\=&\frac{n'!}{(\cb-2)!(n'-\cb+2)!}*2^{n^*}*\alpha\\
=&2^{n^*}*\underbrace{(n'-\cb+3)(n'-\cb+4)\dots (n'-\cb+2+n^*)}_{n^*}\\
=&\underbrace{2*(n'-\cb+3)2*(n'-\cb+4)\dots 2*(n'-\cb+2+n^*)}_{n^*}.
\end{split}
\end{align}
}

We also have that
\begin{align}
\label{eq:ne32}
\begin{split}
{n'+n^*\choose \cb-2}*\alpha&=\frac{(n'+n^*)!}{(\cb-2)!(n'+n^*-\cb+2)!}*\alpha\\
&=\underbrace{(n'+1)(n'+2)\dots (n'+n^*)}_{n^*}.
\end{split}
\end{align}

We first compare $2*(n'-\cb+2+n^*)$ and $n'+n^*$.
Since $\cb\ge \frac{n+1}{2}$, $2\cb\ge n'+n^*+4$, thus we have
$$2n'+2n^*-2\cb+4\le n'+n^*.$$
Therefore, (\ref{eq:ne31})$\le$(\ref{eq:ne32}).
\end{proof}

Finally, we show (3) any guru $i$ ($i\not= i^*$) under $\d$ will not deviate.
It is obviously that $i$ has no incentive to change to delegate to $i^*$ by the dynamics of Algorithm~\ref{algo:NEincomplete}, or to delegate to any delegator by Fact~\ref{fa:todelegate}.
Then by Lemma~\ref{prop:ne3}, $i$ can obtain even lower utility if she changes to delegate to another guru rather than $i^*$.
Hence $i$ will not deviate from $\d$. This concludes the proof.
\end{proof}


\section{Further Details on Experiments}

We start by providing detailed descriptions of the IBRD  (\emph{iterated better response dynamincs}) and OSI (\emph{one-shot interaction}) algorithms we used in our experiments. The pseudi-code for IBRD is described in Algorithm~\ref{algo:better}:  according to a given sequence, each agent repeatedly checks whether delegating to a randomly picked neighbour increases her payoff. In OSI, starting from the trivial profile, all agents simultaneously choose the delegation strategy that maximizes their utility:
let $\d^0$ be the trivial profile, i.e., $\forall i\in N$, $\d^0(i)=i$, for any $i\in N$, $i$ chooses $d_i= \arg\max_{a\in (E(i)\cup\{i\})}u_i((\d^0_{-i},d_i=a))$. 

\begin{algorithm}
\caption{Iterated Better Response Dynamincs (IBRD)}
\label{algo:better}
\begin{description}
\item[Initialization:]
$\d^0: \forall a\in N, \d^0(a)=a$, $\sigma$, $j=1$, $i=1$, $k=0$.
\item[Round $j$:]
\item \textbf{step 1:} Randomly choose $\ell$ from $E(\sigma(i))\cup\{\sigma(i)\}$, and let $\tilde{\d}^j=(\d^j_{-\sigma(i)},\tilde{\d}^j(\sigma(i))=\ell)$.
\item \textbf{step 2:} If $u_{\sigma(i)}(\tilde{\d}^j)>u_{\sigma(i)}(\d^j)$, $\d^j\la \tilde{\d}^j$ and go to step 3,\\
else if $u_{\sigma(i)}(\tilde{\d}^j)\le u_{\sigma(i)}(\d^j)$ and $k>2(|E(\sigma(i))|+1)$, go to step 3,\\
otherwise go to step 1 and $k=k+1$.
\item \textbf{step 3:} If $i<n$, $i=i+1$ and go to step 1;\\
else if $i=30$ and $\d^j=\d^{j-1}$, return $\d^j$;\\
otherwise $i=1$ and go to Round $j+1$.
\end{description}
\end{algorithm}

\medskip

We used several measurements of the profiles output by these two algorithms.
First, to measure the delegation structures of a given profile,
we focused on measures concerning structural properties of the delegation graphs (ratio of delegators, maximum/average length of all delegation chains) and concerning the distribution of \DB\ values (maximum, minimum and mean values, as well as the Gini coefficient of the distribution of \DB\ values). We also computed the weighted average accuracy of all gurus (not mentioned in the paper).

All experiments were implemented in Python 3.
We use functions from the random module\footnote{ https://docs.python.org/3/library/random.html\# } to generate random numbers, and the random seeds were drawn from the system macOS Catalina.

\subsection{Equipment}

We ran instances on two types of machines.
One is MacBook Pro (13-inch), with a cpu of 2,3 GHz Dual-Core Intel Core i5, and memory of 16 GB 2133 MHz LPDDR3.
The other is EC2 server of c5.large type on Amazon Web Service, and it has 2 vCPUs of 3.1GHz and memory of 4GB.
All instances were run on cpus.

\subsection{Additional plots}

Finally, we provide additional plots to those reported in the paper for both experiments A and B and report on a third experiment (experiment C) we only briefly mentioned in the paper.

\paragraph{Experiment A}

\begin{figure}[t]
\centering
\subcaptionbox
{mean longest chains\label{a:long1}}{\includegraphics[width=3.7cm]{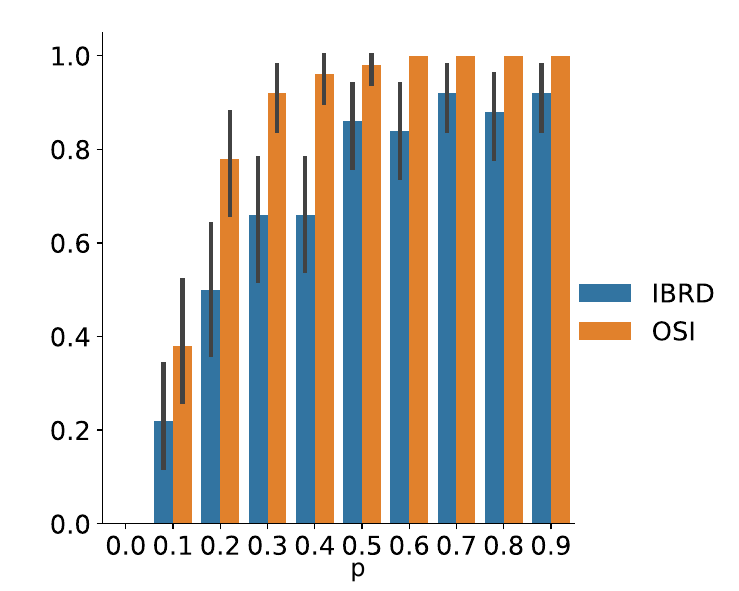}}
\subcaptionbox
{mean average chain lengths \label{a:avglen1}}{\includegraphics[width=3.7cm]{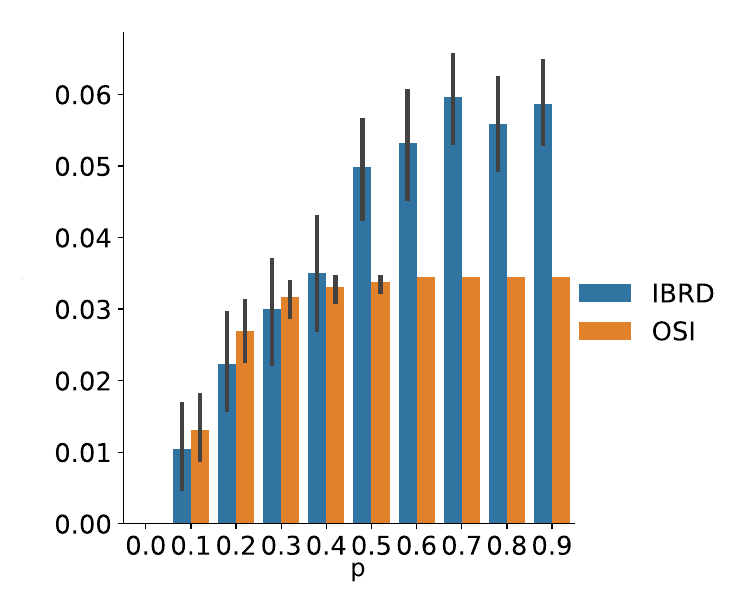}}
\par
\subcaptionbox
{mean maximum $\DB$\label{a:max1}}{\includegraphics[width=3.7cm]{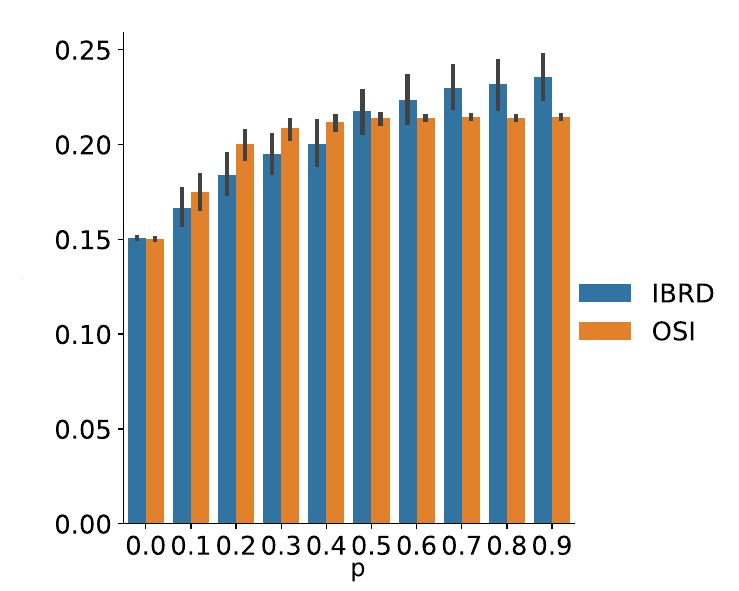}}
\subcaptionbox
{mean minimum $\DB$\label{a:min1}}{\includegraphics[width=3.7cm]{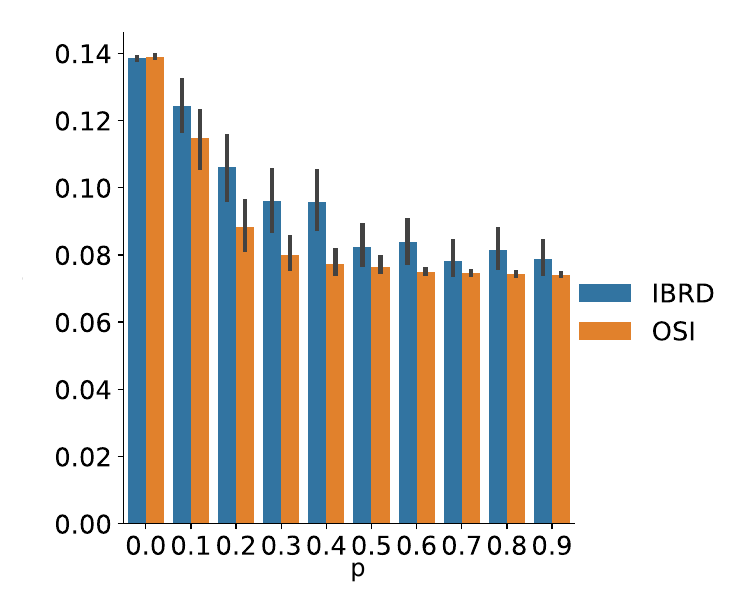}}
\par
\subcaptionbox
{mean accuracy\label{a:acu1}}{\includegraphics[width=3.7cm]{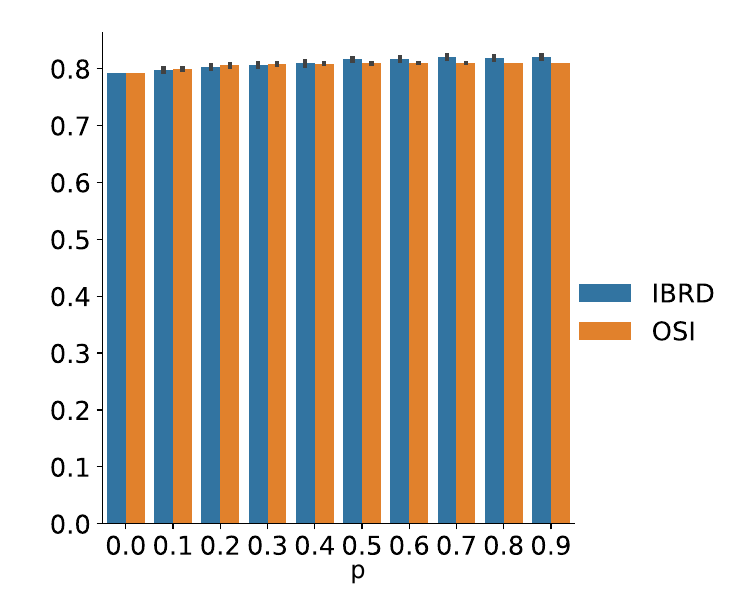}}
\caption{Additional figures for experiment A}
\end{figure}

In addition to the plots in the paper, we show here the mean of the longest chains and the average length of chains for both IBRD and OSI in Fig.~\ref{a:long1} and \ref{a:avglen1}.
Observe that as $p$ becomes larger, both longest chains and average lengths generally become larger, for both IBRD and OSI.
Notice that for both algorithms, the average longest chains are not longer than $1$.
The longest chain of IBRD is averagely shorter than that of OSI, however, the average length of IBRD is larger than that of OSI when $p\ge 0.4$.
We can infer by the above observations that, in OSI, agents choose their strategies based on the trivial profile, and their delegations are not influenced by other agents' strategies.
However, in IBRD, agents can avoid to participate in a long delegation chain because of the multi-round iteration.

We also provide the average maximum/minimum $\DB$s in Fig.~\ref{a:max1} and \ref{a:min1}.
Observe that the maximum $\DB$ increases as $p$ becomes larger, while the minimum $\DB$ decreases.
Recall that when $p>0.4$, the ratio of delegators for IBRD is larger than that for OSI.
The trend of the maximum $\DB$ is consistent: when more agents have access to delegate to high-accuracy agents, those high-accuracy agents are able to collect more delegations, and have more power.
However, the minimum $\DB$ decrease as $p$ becomes larger, which is reverse to the trend of the maximum $\DB$.
Two reasons can contribute to these trends.
One is that the more delegators there are, the more agents lose power as implied in Fact~\ref{fa:todelegate}.
The other reason is that the longest/average length of chains becoming larger also lowers some agents' $\DB$, in accordance with Facts~\ref{fa:powermonotonicity} and \ref{fact:accrual}.

We consider that in delegation games, any rational agent always decides their strategy by approaching a trade-off between accuracy and power.
Then, we also compare the weighted average accuracy of gurus based on the profiles output by these two algorithms (Fig.~\ref{a:acu1}).
We can observe that the average accuracies of IBRD and OSI generally increase as $p$ increases.
The trend is consistent with that of delegator ratio for both algorithms and is in line with that observed in \cite{bloembergen19rational}.

\paragraph{Experiment B}
\begin{figure}[t]
\centering
\subcaptionbox
{mean longest chains\label{b:long1}}{\includegraphics[width=3.7cm]{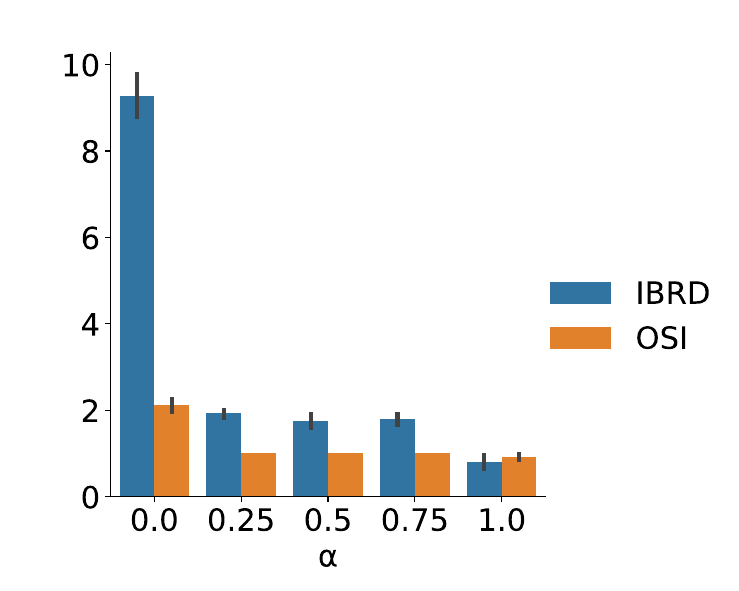}}
\subcaptionbox
{mean average chain lengths \label{b:avglen1}}{\includegraphics[width=3.7cm]{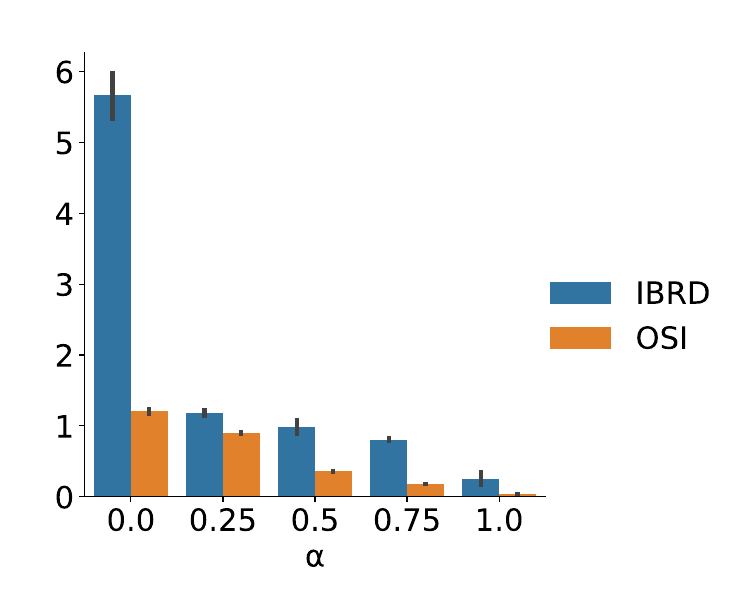}}
\par
\subcaptionbox
{mean maximum $\DB$\label{b:max1}}{\includegraphics[width=3.7cm]{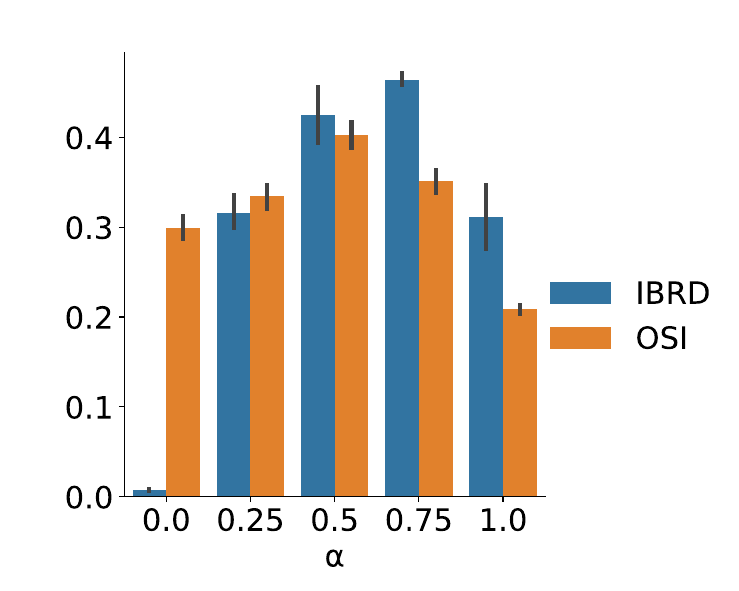}}
\subcaptionbox
{mean minimum $\DB$\label{b:min1}}{\includegraphics[width=3.7cm]{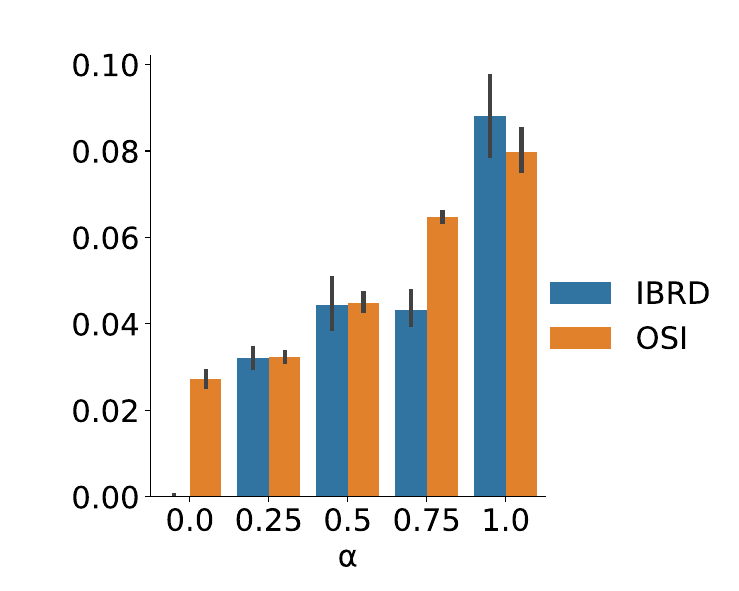}}
\par
\subcaptionbox
{mean accuracy\label{b:acu1}}{\includegraphics[width=3.7cm]{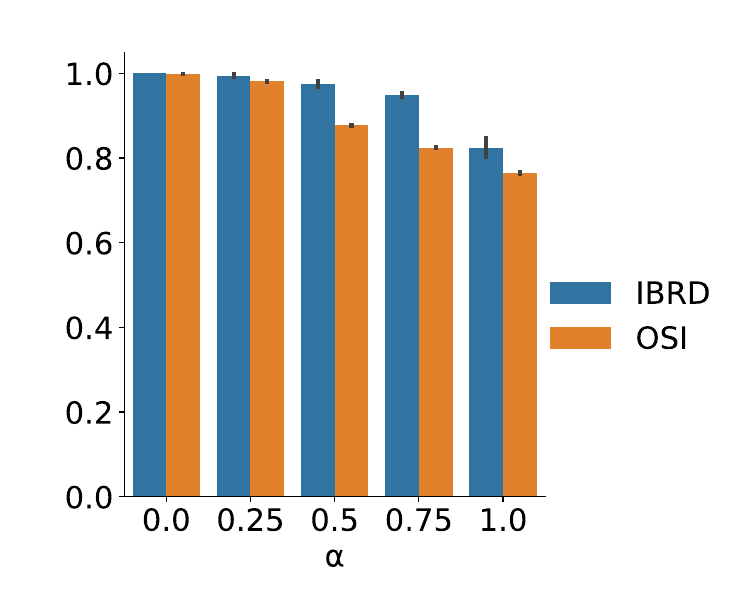}}
\subcaptionbox
{number of converged instances\label{b:con}}{\includegraphics[width=3.7cm]{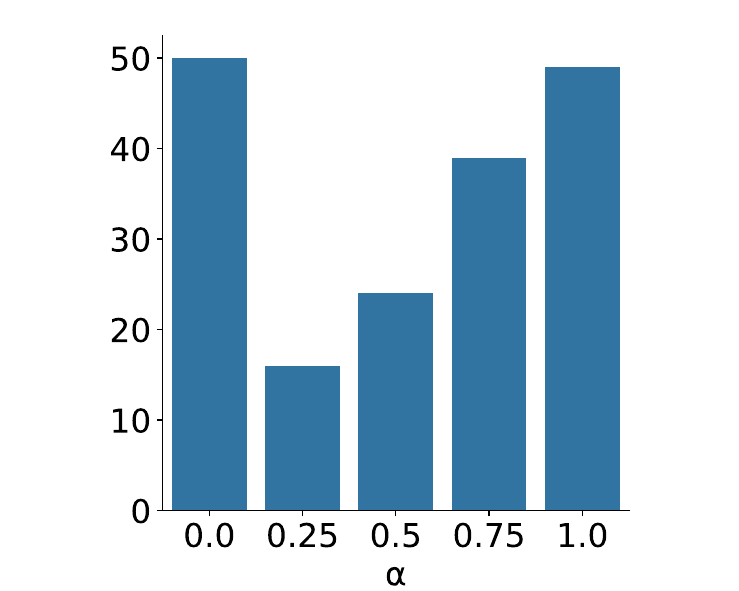}}
\caption{Additional figures for experiment B}
\end{figure}
First, Fig.~\ref{b:long1} and Fig.~\ref{b:avglen1} illustrate the average longest chains and average of average length of chains over 50 instances.
The plots for both IBRD in Fig.~\ref{b:long1} and \ref{b:avglen1} show a decrease as $\alpha$ increases, since delegators far from their gurus have reletively lower $\DB$.
Therefore, as agents attach more weight to power, they have lower incentive to delegate to far gurus.
The plots for OSI also have similar trends, although the reason is different as it only depends on agents' incentive to delegate from the trivial profile.

Then in Fig.~\ref{b:max1} and \ref{b:min1}, we show (average over 50 instances) the maximum and minimum $\DB$s.
Notice that all maximum/minimum $\DB$s of IBRD significantly increase from $A=0$ to $A=0.25$ because of the shortening of the longest/average length of delegation chains due to the weight attached to $\DB$.
The trend coincidences with Fact~\ref{fa:forward}.
Then from $A=0.25$ to $A=0.75$, the ratio of delegators slightly decreases, as well as the longest/average length of delegation chains.
This change makes the maximum and minimum $\DB$s mildly increase.
Finally from $A=0.75$ to $1$, since the delegator ratio drops substantially, the maximum $\DB$ also decreases significantly since gurus lose many delegations.
However, the minimum $\DB$ still increases.

Next, we compare the average accuracies of IBRD and OSI, as shown in Fig.~\ref{b:acu1}.
It can be observed that as more agents delegate, the average accuracy tends to be higher.
The reason is that agents with higher accuracy, and that accrue more delegations, can contribute higher accuracy to the network.
It is worth observing that IBRD outperforms OSI in terms of accuracy because of more agents being delegators in IBRD, and agents having the opportunity to seek (further) gurus with higher accuracies.

Finally, we also depict the number of converged instances over all $\alpha$'s in Fig.~\ref{b:con}.
Since OSI always terminates after one round, we only show the result of IBRD.
When $A=0$, all 50 instances converge but the number substantially decreases to around 17 when $A=0.25$.
When $A\ge 0.25$, the number gradually increases to 50 (when $A=1$).
Recall that the delegator ratio also descends from $A=0.25$ to $1$ by a large extent.
As more delegators exist in the network, the IBRD averagely takes more iterations to converge and the probability of instances that do not converge within 50 iterations in our data set also grows.

\paragraph{Experiment C}

\begin{figure}[t]
\centering
\subcaptionbox
{number of converged IBRD instaces\label{c:con1}}{\includegraphics[width=3.7cm]{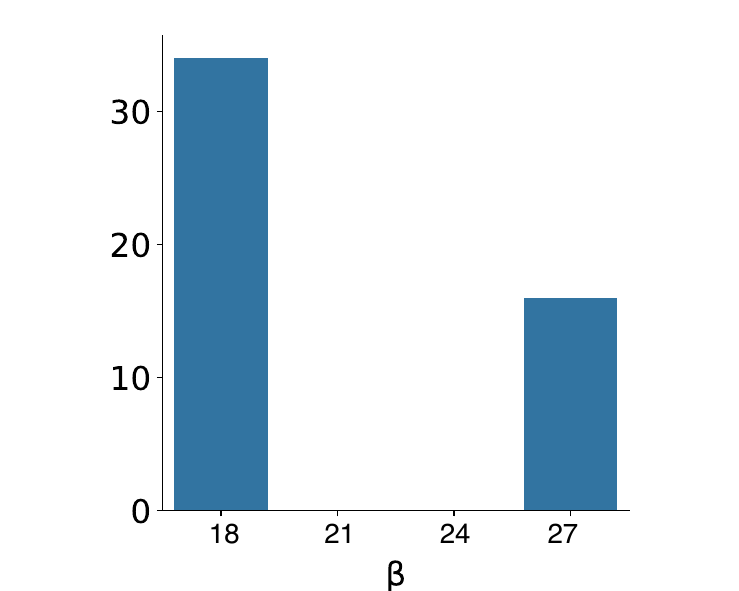}}
\subcaptionbox
{mean delegator ratios \label{c:ratio1}}{\includegraphics[width=3.7cm]{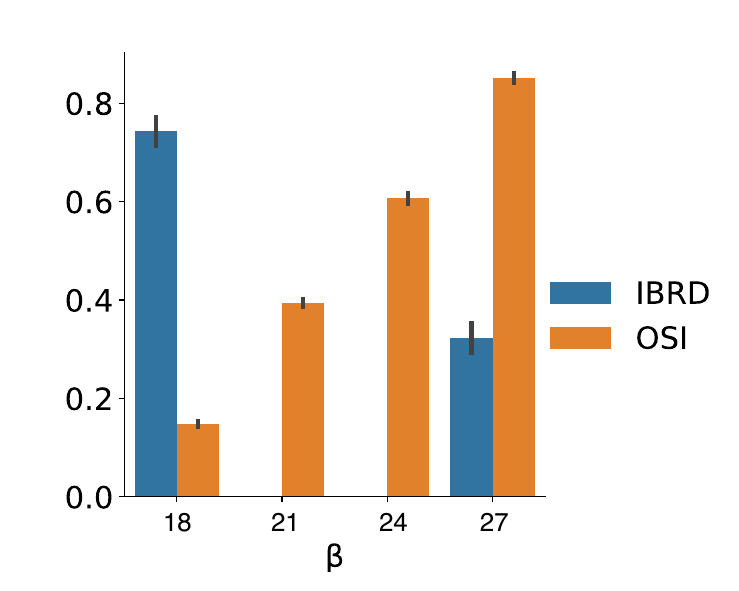}}
\par
\subcaptionbox
{mean average chain lengths\label{c:avglen1}}{\includegraphics[width=3.7cm]{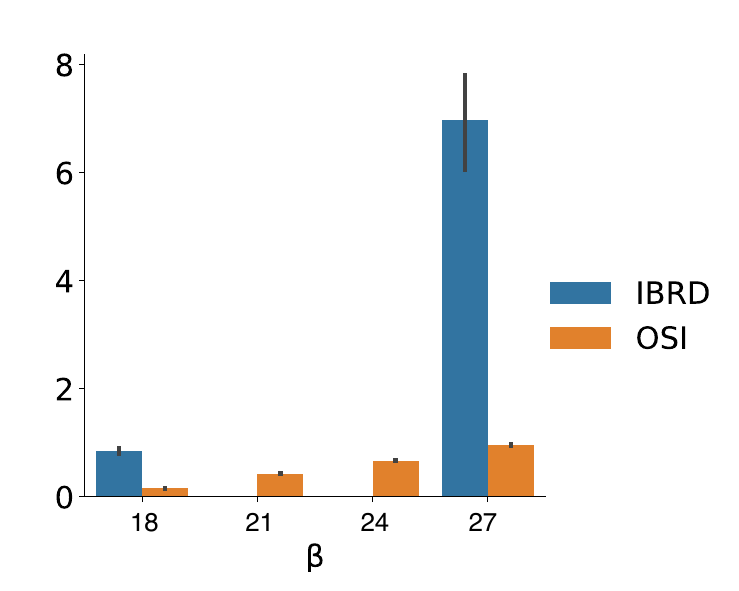}}
\subcaptionbox
{mean longest chains\label{c:long1}}{\includegraphics[width=3.7cm]{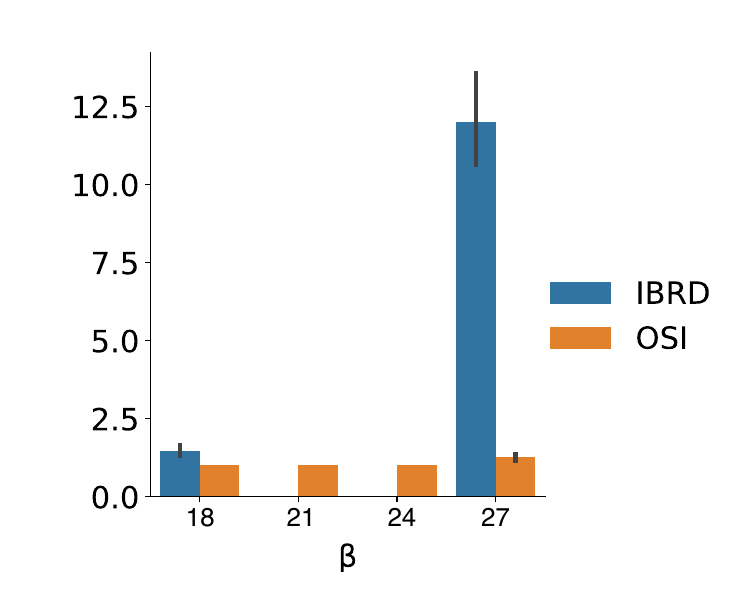}}
\caption{Figures for experiment C}
\end{figure}

In experiment C, we study the agents' behavior by varying the quota $\beta$ in range $\beta=[18,21,24,27]$, that is from majority to unanimity.
However, when $\beta$ is large, especially larger than $\beta=18$, many iterations are needed in order for IBRD to converge as shown in Fig.~\ref{c:con1}. As a result, our data set is based on several instances that have not converged within the stipulated 50 iterations. We can nonetheless also observe some interesting trends in Fig.~\ref{c:ratio1}, \ref{c:avglen1} and \ref{c:long1}.
In Fig.~\ref{c:ratio1}, the delegator ratio of OSI obviously increases linearly as $\beta$ becomes larger.
It implies that agents have more incentive to delegate with larger $\beta$ (and $\beta> n/2$).
The average and longest chains also slightly increase as $\beta$ becomes larger, as shown in Fig.~\ref{c:avglen1} and \ref{c:long1}.
It is also worth mentioning that when $\beta=27$, the average length of delegation chains is around $7$ and the average longest chain is around $12$, which are much longer than that when $\beta=18$.

\clearpage



\end{document}